\documentclass[preprint,3p]{elsarticle}
\makeatletter
\usepackage{setspace}
\usepackage{epsfig}
\usepackage{rotating}
\usepackage{amssymb}
\usepackage{xcolor}
\usepackage[T1]{fontenc}
\usepackage{lmodern}
\usepackage[utf8]{inputenc}
\newenvironment{proof}{\paragraph{Proof:}}{\hfill(Proved)}
\usepackage{booktabs}
\usepackage{multirow}
\usepackage{enumerate}
\usepackage{caption}
\usepackage{subcaption}
\usepackage{lscape,float}
\usepackage{graphicx}
\usepackage[superscript,biblabel]{cite}
\usepackage{xcolor}
\usepackage{soul}
\usepackage[toc]{appendix}

\usepackage{natbib,hyperref,doi}
\usepackage[superscript,biblabel]{cite}
\setcitestyle{authoryear,round}
\makeatletter
\renewcommand\@biblabel[1]{#1.}
\makeatother
\usepackage{algorithm}
\usepackage[noend]{algpseudocode}

\usepackage[fleqn]{amsmath}

\usepackage{gensymb}
\usepackage{mathtools}
\newtheorem{theorem}{Theorem}

\newtheorem{lemma}{Lemma}
\makeatletter
\renewcommand\@biblabel[1]{#1.}
\makeatother

\begin{document}
\newcommand{\ve}[1]{\mbox{\boldmath$#1$}}
\newcommand{\IR}{\mbox{$I\!\!R$}}
\newcommand{\1}{\Rightarrow}
\newcommand{\bs}{\baselineskip}
\newcommand{\esp}{\end{sloppypar}}
\newcommand{\be}{\begin{equation}}
\newcommand{\ee}{\end{equation}}
\newcommand{\beano}{\begin{eqnarray*}}
\newcommand{\inp}[2]{\left( {#1} ,\,{#2} \right)}
\newcommand{\eeano}{\end{eqnarray*}}
\newcommand{\bea}{\begin{eqnarray}}
\newcommand{\eea}{\end{eqnarray}}
\newcommand{\ba}{\begin{array}}
\newcommand{\ea}{\end{array}}
\newcommand{\nno}{\nonumber}
\newcommand{\dou}{\partial}
\newcommand{\bc}{\begin{center}}
\newcommand{\ec}{\end{center}}
\newcommand{\2}{\subseteq}
\newcommand{\cl}{\centerline}
\newcommand{\ds}{\displaystyle}

\def\refhg{\hangindent=20pt\hangafter=1}
\def\refmark{\par\vskip 2.50mm\noindent\refhg}

\title{Reliability Acceptance Sampling Plans under Progressive Type-I Interval Censoring Schemes in Presence of Dependent Competing Risks}
\author[label1]{Rathin Das\corref{cor1}}
\cortext[cor1]{Corresponding author}
\ead{rathindas65@gmail.com}
\author[label2]{Soumya Roy}
\author[label1]{Biswabrata Pradhan}
\address[label1]{SQC and OR Unit, Indian Statistical Institute, Kolkata, India}
\address[label2]{Indian Institute of Management Kozhikode, Kozhikode, Pin 673570, India}


\date{}

\begin{abstract}
We discuss the development of reliability acceptance sampling plans under progressive Type-I interval censoring schemes in the presence of competing causes of failure. We consider a general framework to accommodate the presence of independent or dependent competing risks and derive the expression for the Fisher information matrix under this framework. We also discuss the asymptotic properties of the maximum likelihood estimators, which are essential in obtaining the sampling plans. Subsequently, we specialize in a frailty model, which allows us to accommodate the dependence among the potential causes of failure. The frailty model provides an independent competing risks model as a limiting case. We then present the traditional sampling plans for both independent and dependent competing risks models using producer’s and consumer’s risks. We also consider the design of optimal PIC-I schemes in this context and use a c-optimal design criterion, which helps us to obtain more useful reliability acceptance sampling plans in the presence of budgetary constraints. We conduct a comprehensive numerical experiment to examine the impact of the level of dependence among the potential failure times on the resulting sampling plans. We demonstrate an application of the developed methodology using a real-life example and perform a simulation study to study the finite sample properties of the developed sampling plans. The methodology developed in this article has the potential to improve the design of optimal censoring schemes in the presence of competing risks while taking into account budgetary constraints.   
\end{abstract}

\maketitle

\section{Introduction}\label{INTRO}
 \indent In the present-day business world, characterized by intense competition, prioritizing the reliability of products is paramount for manufacturers. Although reliability is an important consideration during the product development phase, achieving the target reliability levels throughout the post-production stages is equally critical. A standard approach is to achieve this through life tests \citep[pp.~2--3]{ME_98}. However, these tests are often performed under various censoring schemes, resulting in incomplete data sets. Among various censoring schemes, Type-I and Type-II schemes are perhaps the two most widely adopted schemes, primarily because of their inherent simplicity and ease of implementation.  These traditional schemes require continuous inspection during the experiment, which may not be feasible due to insufficient resources. In view of this, interval censoring schemes are proposed in the literature, resulting in grouped failure data sets \citep{COUSINO2024110424}. \\
\indent A major drawback of these interval censoring schemes is their inability to accommodate intermediate withdrawals of test units during the course of an experiment. However, such intermediate withdrawals may be unavoidable for various practical reasons \citep{LEE2024109843}. For example, consider the life test performed on 68 battery cells, as discussed by \citet[][p.~633]{ME_98}. Some battery cells are withdrawn from the ongoing experiment to assess the extent of physical degradation. To accommodate such intermediate withdrawals of test units from the life test, progressive Type-I interval censoring (henceforth, PIC-I) schemes are proposed in the reliability literature \citep{aggar_01}. \\
\indent In the existing literature, there is a considerable number of work on statistical inference of various important lifetime models based on such PIC-I data sets. See, for example, \citet{aggar_01, Ng_09, chen_10, roy_17, roy_ijqrm_21}. 
These works mainly focused on the items that have a single failure mode and discussed classical and Bayesian inference of lifetime parameters based on PIC-I data sets. However, as in the case of the battery cell data set mentioned above, an item often experiences multiple failure modes, causing it to fail when one of these failure modes occurs. These failure modes are commonly referred to as competing risks in the literature \citep{SARHAN2010953}. There is limited literature on statistical inference of lifetime data sets under PIC-I schemes in the presence of competing risks. See \citet{ROY_BP_23} for a more detailed account of related literature in this regard. \\
\indent However, an important issue in this context is the optimal planning of life-test experiments under these PIC-I schemes. Recently, this has attracted a lot of interest in the literature. \citet{lin_09} and \citet{lin_11} provided $A$- and $D$-optimal designs for log-normal and Weibull models, respectively. \citet{roy_17} presented Bayesian $D$-optimal design for the log-normal lifetime model. This was further extended to Bayesian $c$-optimal designs by \citet{roy_19} for Weibull and log-normal lifetime models. \citet{sonal_17} provided a cost function-based approach in this context. These works mainly focused on determining the optimal design parameters, such as optimal inspection times and the corresponding number of test unit withdrawals from the experiment. However, a more pertinent question in this context may be the design of optimal reliability acceptance sampling plans (henceforth, RASPs) \citep{ZHENG2023108877}. RASPs are essentially variations of variable acceptance sampling plans, where we use the reliability of a product as a measurable quality characteristic of interest \citep{FERNANDEZ2012719,BHATTACHARYA201591}. \citet{huang_08} presented a cost function-based approach to obtain the RASPs for the exponential lifetime model. \citet{Wu_14} further extended their work in the competing risks framework, assuming that potential failure times follow independent exponential lifetime distributions. \\
\indent In this article, we present the RASPs under PIC-I schemes in the presence of competing risks. Unlike existing work, we develop a general framework for modeling the PIC-I data set in the presence of competing risks. This framework allows us to incorporate the presence of both independent and dependent competing risks. Potential failure times under competing failure modes are typically assumed to be independent in the existing literature. However, these failure times are often correlated in practice (see, for example, \citet{ZHANG2024109796,LIU_2012_TM}).  Subsequently, we obtain the Fisher information matrix under this general framework. The Fisher information matrix depends only on the first-order partial derivatives, which reduces the computational complexity to a great extent. We also present the asymptotic properties of the maximum likelihood estimators (henceforth, MLEs) under a set of regularity conditions. These asymptotic properties of the MLEs are essential for the development of subsequent RASPs.\\
\indent We then consider a frailty-based dependent competing risks model as a special case of the general model presented in this article \citep{Hougaard_1995}. A frailty model is a type of random effects model for failure times. As is standard in the existing literature, we assume that the random frailty variable exerts a multiplicative effect on the hazard rates. We further consider a gamma distribution for the frailty variable and assume that the potential failure times under competing failure modes are conditionally independent, given the frailty variable. A major advantage of this frailty-based dependent competing risks model is that it allows us to obtain an independent competing risks model as a special case. Subsequently, we present the optimal RASPs, providing the optimal sample size and a lower acceptance limit under both independent and dependent competing risks frameworks. We also present the popular $c$-optimal RASPs under PIC-I schemes. We then introduce a generic cost function and develop an approach to obtain the $c$-optimal RASPs under a cost constraint.  We perform a detailed numerical experiment to evaluate the effect of incorporating the dependence among the competing failure modes on the resulting RASPs. We also provide an example to illustrate an application of the developed methodology. Furthermore, we conduct a simulation exercise to study the finite-sample properties of the optimal RASPs.\\
\indent Thus, this article contributes to the existing literature in three important directions. First, we present a comprehensive framework to model PIC-I datasets while accounting for the presence of competing failure modes.  This framework allows us to incorporate the dependence structure among the competing failure modes, which may be useful for modeling the interval-censored data sets often found in reliability applications. Second, we provide a generic expression of the Fisher information matrix and also derive the asymptotic properties of the MLEs under a set of regularity conditions. These results may be useful for performing statistical inference based on interval-censored data sets. Third, we develop the RASPs under three practical scenarios and highlight the importance of incorporating the dependence among the potential failure times in the presence of competing failure modes.\\
\indent The remainder of this article is organized as follows. The general framework for modeling PIC-I data sets is provided in Section \ref{gen_frame}. Also, the asymptotic properties of the MLEs are presented in Section \ref{gen_frame}. Subsequently, the frailty model is considered a special case in Section \ref{Frail_Mod}. The optimal RASPs under the PIC-I schemes are developed in Section \ref{dev}. The methodology for obtaining traditional $c$-optimal sampling plans with and without cost constraints is discussed in Section \ref{OPT}. Subsequently, the effect of incorporating the dependence among the competing failure modes with a detailed numerical experiment is examined in Section \ref{num_exp}. An application of the developed methodology using a real-life example is presented in Section \ref{Exam}. A simulation study is undertaken to investigate the finite sample properties of the developed RASPs in Section \ref{Exam}. Finally, the contributions of this work and outline of potential future extensions are highlighted in Section \ref{con}.

\section{General Framework}\label{gen_frame}
\indent In this section, we first present the general model and likelihood function in \S \ref{mod}. 
Subsequently, we discuss the asymptotic properties of the MLEs in \S \ref{Fisher}. 
\subsection{Model and Likelihood Function}\label{mod}
\indent Suppose that in a life test experiment, a test unit fails as soon as one of the competing failure modes $J$ occurs. Let $X_j$ denote the potential failure time due to the $j$th cause of failure for $j=1,\ldots, J$. Let $T$ denote the overall failure time of the test unit, and the index $C$ denote the cause of failure. It is easy to understand that $T=\min\{X_1,\ldots, X_J\}=X_C$.  Let $f_{X_j}(x\ |\ \boldsymbol{\theta}_j)$ and ${F}_{X_j}(x\ |\ {\boldsymbol{\theta}_j})$ be the probability density function (henceforth, p.d.f.) and cumulative distribution function (henceforth, c.d.f.) of $X_j$, respectively, where $\boldsymbol{\theta}_j$ is the vector of parameters corresponding to the $j$th cause of failure. Also, denote the joint reliability function (henceforth, r.f.) of $\boldsymbol{X}=(X_1,\ldots, X_J)$ by $\overline{F}_{\boldsymbol{X}}(\boldsymbol{x}\ |\ \boldsymbol{\theta})$, where $\boldsymbol{\theta}=(\theta_1,\ldots,\theta_s)$ is a vector of parameters, with $s$ being the number of model parameters.  \\
\indent In a competing risk setup, we collect data on $C$ and $T$ from a life test. Typically, the joint distribution of $C$ and $T$ is specified in terms of the sub-distribution function $G(j,t\ |\ {\boldsymbol{\theta}})=P(C=j,T\le t)$ or equivalently, by the sub-survivor function $\overline{G}(j,t\ |\ {\boldsymbol{\theta}})=P(C=j,T> t)$ \citep[Chapter 3]{crowder2001classical}. The sub-density function corresponding to the observed pair $\{C=j, T=t\}$ is then given by
\allowdisplaybreaks\bea
g(j,t\ |\ {\boldsymbol{\theta}})&=&\left[-\frac{\delta \overline{F}_{\boldsymbol{X}}\left(\boldsymbol{x}\ |\ \boldsymbol{\theta}\right)}{\delta x_j}\right]_{t\boldsymbol{1}_J},\nonumber
\eea
where $\boldsymbol{1}_J$ is a unit vector of $J$ components and  $[\ldots]_{t\boldsymbol{1}_J}$ signifies that the enclosed partial derivative is to be computed at $t\boldsymbol{1}_J=(t,\ldots,t)$ for $j=1,\ldots,J$. Also, let the p.d.f., c.d.f. and r.f. of $T$ be given by  $f_T(t\ |\ {\boldsymbol{\theta}})$, $F_T(t\ |\ {\boldsymbol{\theta}})$ and $\overline{F}_T(t\ |\ {\boldsymbol{\theta}})$, respectively. \\
\indent Now suppose that $n$ such identical test units having $J$ competing failure modes are subjected to a life test under a $M$-point PIC-I scheme. In a $M$-point PIC-I scheme, the units are monitored at the pre-decided inspection times $L_1<L_2<\ldots<L_M$, where $L_M$ is the pre-fixed termination time of the experiment. Let $N_i$ denote the number of units that are at risk of failure at the start of the $i$-th interval $(L_{i-1}, L_i]$ and $D_{ij}$ be the number of failures occurring in this interval due to the $j$th cause of failure for $i=1,\dots, M$ and $j=1,\ldots, J$. Note that these failures are observed only at $L_i$. Furthermore, in a $M$ point PIC-I scheme, $R_i$ surviving items are removed from the life test at each intermediate time $L_i$, for $i=1,2\ldots, M-1$. Subsequently, all remaining test units, say $R_M$, are removed from the life test at $L_M$.  For convenience, let $\boldsymbol{D}_M=(D_{11},\ldots, D_{1J},\ldots, D_{M1},\ldots, D_{MJ})$ and $\boldsymbol{R}_M=(R_1,\ldots, R_M)$ denote the number of failures and number of withdrawals for the $M$-point PIC-scheme, respectively.\\ 
\indent It is easy to see that $N_1=n$ and $N_i=N_{i-1}-\sum_{j=1}^J D_{i-1,j}-R_{i-1}$, for $i= 2,\ldots, M$. Note that $R_i$s should not be greater than $N_i$s. Furthermore, $R_i$s are typically specified as the proportion of surviving units at $L_i$.  Thus, for $i=1,\ldots, M-1$, $R_i=\lfloor p_i \times (N_{i}-\sum_{j=1}^J D_{ij})\rfloor$, where $p_i$ is the withdrawal proportion, with $\lfloor x\rfloor$ denoting the greatest integer less than or equal to $x$. It is obvious that $p_M=1$. If $p_i=0$ for $i=1,\ldots, M-1$, the PIC-I scheme reduces to a traditional interval censoring scheme. Thus, in order to conduct a life test under a PIC-I scheme, we must decide $n$, $M$, $\boldsymbol{L}_M$ and $\boldsymbol{p}_M$, where $\boldsymbol{L}_M=(L_1,\ldots, L_M)$ and $\boldsymbol{p}_M=(p_1,\ldots, p_{M-1},p_M=1)$ represent the sequence of inspection times and corresponding withdrawal proportions for the $M$-point PIC scheme, respectively. For convenience, we denote these decision variables by $\boldsymbol{\boldsymbol{\zeta}} = (n, M, \boldsymbol{L}_M, \boldsymbol{p}_M)$.\\
\indent Let $q_{ij}$ be the probability that a test unit fails in the $i$th interval due to the $j$th failure mode. It is easy to see that
\begin{align}\label{qij}
    q_{ij}=\frac{P\left(C=j,T\le L_i\right)-P\left(C=j, T\le L_{i-1}\right)}{1-P\left(T\le L_{i-1}\right)}=&\frac{G\left(j,L_{i}\ | \ {\boldsymbol{\theta}}\right)-G\left(j,L_{i-1}\ | \ {\boldsymbol{\theta}}\right)}{\overline{F}_T(L_{i-1}\ | \ {\boldsymbol{\theta}})}\nonumber\\
    =&\frac{\overline{G}\left(j,L_{i-1}\ | \ {\boldsymbol{\theta}}\right)-\overline{G}\left(j,L_{i}\ | \ {\boldsymbol{\theta}}\right)}{\overline{F}_T(L_{i-1}\ | \ {\boldsymbol{\theta}})}.
\end{align}
Furthermore, let $q_i$ represent the probability that a test unit fails in the $i$th interval. Then, we have
\begin{align}\label{qi}
q_i=\frac{P\left(T \le L_{i}\right)-P\left(T \le L_{i-1}\right)}{1-P\left(T \le L_{i-1}\right)}&=\frac{F_T(L_{i}\ | \ {\boldsymbol{\theta}})-F_T\left(L_{i-1}\ | \ {\boldsymbol{\theta}}\right)}{\overline{F}_T(L_{i-1}\ | \ {\boldsymbol{\theta}})}\nonumber\\ &=\frac{\overline{F}_T(L_{i-1}\ | \ {\boldsymbol{\theta}})-\overline{F}_T(L_{i}\ | \ {\boldsymbol{\theta}})}{\overline{F}_T(L_{i-1}\ | \ {\boldsymbol{\theta}})}.   
\end{align}
Note that $\sum_{j=1}^J q_{ij}=q_i$. Now, it is easy to see that
$$\left(D_{i1},\ldots, D_{iJ}\right)|\boldsymbol{D}_{i-1},\boldsymbol{R}_{i-1} \sim \text{multinomial} \left(N_i, q_{i1},\ldots, q_{iJ},1-q_i\right).$$
Then, we can recursively obtain the expressions for $E[N_i]$, $E[{D_{ij}}]$ and $E[R_i]$, which are are as follows:

\begin{align}
    E\left[N_i\right]=n\prod_{l=0}^{i-1} (1-q_l)(1-p_l),\nonumber
\end{align}
\begin{align}
E\left[{D_{ij}}\right]=n \left\{\prod_{l=0}^{i-1} \left[(1-q_l)(1-p_l)\right]\right\} q_{ij},\nonumber  
\end{align}
and
\begin{align}
    E\left[R_{i}\right]=n \left\{\prod_{l=0}^{i-1} \left[(1-q_l)(1-p_l)\right]\right\} (1-q_i)p_{i},\nonumber
\end{align}
for $i=1,\ldots, M$ and $j=1,\ldots,J$, with $p_0=0$ and $q_0=0$ (see \citep{roy_19, Wu_14} for details). Note that $D_{i+}=\sum_{j=1}^J D_{ij}$ is the number of failures that occur in the time interval $(L_{i-1}, L_i]$ and $D=\sum_{i=1}^MD_{i+}$ is the total number of failures in the life test.  Thus, we have
\begin{align}\label{dsum}
    E[D]=&\sum_{i=1}^M E[D_{i+}],
\end{align}   
where
\begin{align} 
    E\left[{D_{i+}}\right]=n \left\{\prod_{l=0}^{i-1} \left[(1-q_l)(1-p_l)\right]\right\}q_i. \nonumber
\end{align}


\indent We now briefly introduce the observed data from a $M$-point PIC-I scheme. Let $n_i$ be the observed value of $N_i$, for $i=1,\ldots, M$. Furthermore, let $\delta_{ikj}$ represent an indicator function that takes value 1, if $k$th surviving test unit fails in the $i$-th interval $(L_{i-1}, L_i]$ due to the $j$th failure mode, and 0, otherwise, for $k=1, \ldots, n_i$, and $j=1,2,\ldots, J$. Then, $\sum_{k=1}^{n_i}\delta_{ikj}=d_{ij}$, where $d_{ij}$ is the observed value of $D_{ij}$. Let 
$$\mathcal{D}_M=\{(d_{11},d_{12},\ldots,d_{1J},r_1),(d_{21},d_{22},\ldots,d_{2J},r_2),\ldots,(d_{M1},d_{M2},\ldots,d_{MJ},r_M)\}$$
denote the observed data under the PIC-I scheme, where $r_i$ is the observed value of $R_i$, for $i=1,\ldots,M$.\\
\indent Based on the observed data $\mathcal{D}_M$, it can be shown that the likelihood function of $\boldsymbol\theta$ is given by
\begin{align*}
L(\boldsymbol{\theta}\ | \ \mathcal{D}_M,\boldsymbol{\boldsymbol{\zeta}})\propto\prod_{i=1}^M\prod_{k=1}^{n_i}\left[\prod_{j=1}^J q_{ij}^{\delta_{ikj}}\right](1-q_i)^{1-\delta_{ik+}},
\end{align*}
where $\delta_{ik+}=\sum_{j=1}^J\delta_{ikj}$ and $\sum_{k=1}^{n_i}(1-\delta_{ik+})=n_i-d_{i+}$, where $d_{i+}$ is the observed value of $D_{i+}$. 
The log-likelihood function, after ignoring the proportionality constant, can be written as
\begin{align}
 l(\boldsymbol{\theta}\ | \ \mathcal{D}_M,\boldsymbol{\boldsymbol{\zeta}})&=\log L(\boldsymbol{\theta}\ | \ \mathcal{D}_M,\boldsymbol{\boldsymbol{\zeta}})\nonumber\\
&=\sum_{i=1}^M\sum_{k=1}^{n_i}\left[\sum_{j=1}^J{\delta_{ikj}}\log(q_{ij})+(1-{\delta_{ik+}})\log(1-q_i)\right].  \label{log_like_eqn}
\end{align}
We can easily obtain the MLEs of the model parameters by numerically maximizing the log-likelihood function $l(\boldsymbol{\theta}\ | \ \mathcal{D}_M,\boldsymbol{\boldsymbol{\zeta}})$.
\subsection{Asymptotic Properties of MLEs} \label{Fisher}
\indent We now present the asymptotic properties of the MLEs, which are essential in the development of RASPs. Let $\widehat{\boldsymbol{\theta}}$ denote the MLE of $\boldsymbol{\theta}$. Toward this end, we state a set of regularity conditions required to establish the asymptotic properties of $\widehat{\boldsymbol{\theta}}$. \\
\textit{\textbf{Regularity Conditions}} 
\begin{enumerate}[I.]
    \item The unobserved lifetimes $T_1,T_2,\ldots,T_n$ are independent and identically distributed (henceforth, i.i.d.) with common c.d.f. $F_T(.\ | \ {\boldsymbol{\theta}})$ and p.d.f. $f_T(.\ | \ {\boldsymbol{\theta}})$.
    \item The number of causes of failure of a test unit is finite.
\item  $(T_1,C_1),(T_2,C_2),\ldots,(T_n,C_n)$ are i.i.d. with common joint sub-survivor function $G(\cdot,\cdot\ | \ {\boldsymbol{\theta}})$ and sub-density function $g(\cdot,\cdot\ | \ {\boldsymbol{\theta}})$, where $C_j$ is the cause of the failure of the $j^{th}$ item for $j=1,\ldots,n$.
\item The supports of $f_T$ and $g$ are independent of $\boldsymbol{\theta}$.
    \item The parameter space $\boldsymbol{\Theta}$ contains an open set $\boldsymbol{\Theta}_0$ of which the true parameter $\boldsymbol{\theta}^0$ is an interior point.
    \item For almost all $t$ and for all $j=1,2,\ldots,J$, both $F_T(t\ | \ \boldsymbol{\theta})$ and $G(j,t\ | \ \boldsymbol{\theta})$ admit all third-order derivatives, $\frac{\partial^3 F_T(t\ | \ \boldsymbol{\theta})}{\partial\theta_u\partial\theta_v\partial\theta_w}$ and $\frac{\partial^3 G(j,t\ | \ \boldsymbol{\theta})}{\partial\theta_u\partial\theta_v\partial\theta_w}$, respectively, for all $\boldsymbol{\theta}\in \boldsymbol{\Theta}_0$ and $u,v,w=1,2,\ldots,s$. In addition, all first, second, and third order derivatives of $F_T(t\ | \ \boldsymbol{\theta})$ and $G(j,t\ | \ \boldsymbol{\theta})$ with respect to the parameters are bounded for all $\boldsymbol{\theta}\in\boldsymbol{\Theta}_0$.
    \item $L_i$s are fixed in such a way that
    \begin{enumerate}[(a)]
        \item $0<q_{ij}<1$ and $0<q_i<1$ for $i=1,2,\ldots,M$ and $j=1,2\ldots,J$.
        \item $\nabla_{\boldsymbol{\theta}}\boldsymbol{q}$ is a matrix of rank $s$, where $\nabla_{\boldsymbol{\theta}}\boldsymbol{q}=\left(\frac{\partial q_{i}}{\partial\theta_u}\right)_{s\times M}$ for $u=1,2,\ldots,s$ and $i=1,2,\ldots,M$.
    \end{enumerate}
    \end{enumerate}
\indent The regularity conditions (I)-(VI) presented above are essentially standard requirements for establishing the asymptotic properties of the MLEs in the competing risks set up \citep[see, for example,][] {MUKHOPADHYAY2006803}, whereas condition (VII) ensures the positive definiteness of the Fisher information matrix. We now present the following results, which are required to obtain the Fisher information matrix.
\begin{lemma}\label{l1}
    The first derivative of the log-likelihood satisfies the following:
    \begin{align*}
        E\left[\frac{\partial l(\boldsymbol{\theta}\ | \ \mathcal{D}_M,\boldsymbol{\boldsymbol{\zeta}})}{\partial \theta_u}\right]=0,\,\,\text{for}\,\, u=1,\ldots,s.
    \end{align*}
\end{lemma}
\begin{proof}
  The proof of the lemma is given  in the \ref{l1app}.  
\end{proof}

\begin{lemma} \label{lemma_2}
    The second derivative of the log-likelihood function satisfies the following
    \begin{align*}
        E\left[ \frac{\partial^2 l(\boldsymbol{\theta}\ | \ \mathcal{D}_M,\boldsymbol{\boldsymbol{\zeta}})}{\partial\theta_u\partial\theta_v}\right]=E\left[ \left(\frac{\partial l(\boldsymbol{\theta}\ | \ \mathcal{D}_M,\boldsymbol{\boldsymbol{\zeta}})}{\partial\theta_u}\right)\left(\frac{\partial l(\boldsymbol{\theta}\ | \ \mathcal{D}_M,\boldsymbol{\boldsymbol{\zeta}})}{\partial\theta_v}\right)\right], \,\,\text{for}\,\, u,v=1,\ldots,s.
    \end{align*}
\end{lemma}
\begin{proof}
  The proof of the lemma is given in \ref{l2app}.
\end{proof}

 Using Lemmas \ref{l1} and \ref{lemma_2}, we readily obtain the Fisher information matrix, which is provided in the following theorem.
\begin{theorem}\label{theorem1}
        The Fisher information matrix for $\boldsymbol{\theta}$ is given by 
    \begin{align*}
        I(\boldsymbol{\theta}\ | \ \boldsymbol{\boldsymbol{\zeta}})=\sum_{i=1}^M\left[\sum_{j=1}^J\frac{E[N_{i}]} {q_{ij}}\nabla_{\boldsymbol{\theta}}(q_{ij})\nabla_{\boldsymbol{\theta}}(q_{ij})^T+\frac{E[N_{i}]}{(1-q_{i})}\nabla_{\boldsymbol{\theta}}(q_{i})\nabla_{\boldsymbol{\theta}}(q_{i})^T\right],
    \end{align*}
\end{theorem}
\indent It is interesting to note that the Fisher information matrix involves only the first-order partial derivatives, thus simplifying its computation to a great extent, as illustrated in the subsequent sections. We now present the asymptotic properties of the MLEs $\widehat{\boldsymbol{\theta}}$ in the following theorem.
\begin{theorem}\label{th2}
    If the regularity conditions (I)-(VII) (as stated above) hold, then
    \begin{enumerate}[(I)]
        \item $\widehat{\boldsymbol{\theta}}$ is consistent for $\boldsymbol{\theta}$.
    \item  $\widehat{\boldsymbol{\theta}}$ is asymptotically normal with mean $\boldsymbol{\theta}^0$ and variance-covariance matrix $[I(\boldsymbol{\theta}^0\ |  \ \boldsymbol{\boldsymbol{\zeta}})]^{-1},$ where $\boldsymbol{\theta}^0$ denotes the true value $\boldsymbol{\theta}$. 
    \end{enumerate}
\end{theorem}
\begin{proof}
    The proof of the theorem is given in the \ref{th2app}.
\end{proof}
\section{Frailty Model}\label{Frail_Mod}
\indent In this section, we consider a dependent competing risks model as a special case of the generic model presented above, primarily due to two reasons. First, as discussed in Section \ref{INTRO}, one of the main objectives of this article is to develop RASPs under PIC-I schemes in the presence of competing risks. This necessitates us to assume a specific parametric model, as discussed below. Second, the dependent competing risks model allows us to incorporate the dependence among the potential failure times. The model considered here allows us to obtain an independent competing risks model as a special case. Thus, we can also study the relevance of incorporating the dependence structure among the competing failure modes while deciding the RASPs.\\
\indent We assume that $X_j$, the potential failure time due to the $j$th failure mode, follows a Weibull model, having the p.d.f.
$$f_{X_j}(x\ | \ \boldsymbol{\theta}_j)=\left(\frac{\gamma_j}{\eta_j}\right)\left(\frac{x}{\eta_j}\right)^{\gamma_j-1}\exp\left[-\left(\frac{x}{\eta_j}\right)^{\gamma_j}\right],\,\,\, \mbox{}x>0,$$
where $\boldsymbol{\theta}_j=(\eta_j,\gamma_j)$, with $\eta_j\, (>0)$ and $\gamma_j\, (>0)$ being the scale and shape parameters, respectively, for $j=1,\ldots, J$. As in \citet{Roy_14}, we further assume $\gamma_1=\cdots=\gamma_J=\gamma$, where $\gamma$ is the common shape parameter. Toward this end, we should also note that this assumption often holds true for real data sets (see, for example, the numerical illustration provided by \citep{LIU_2012_TM}). Then, the r.f. of $X_j$ is given by
$$\overline{F}_{X_j}(x\ | \ \boldsymbol{\theta}_j)=\exp\left[-\left(\frac{x}{\eta_j}\right)^\gamma\right]=\exp\left(-\Lambda_j(x)\right),$$
where $\Lambda_j(\cdot)$ is the cumulative hazard function corresponding to the $j$th failure mode. \\
\indent A common approach to introduce the dependence between the potential failure times $X_j$s is through the unobserved random factor $Z$, commonly known as the frailty in the existing literature \citep[][pp.~42-45]{crowder2001classical}. Thus, we assume that $\Lambda_1(x),\ldots,\Lambda_J(x)$ share common random frailty $Z$ and substitute $\Lambda_j(x)$ by $z\Lambda_j(x)$. We further assume that the potential failure times $X_j$s are conditionally independent, given $Z$. Thus, given $Z=z$, the conditional joint r.f. of $\boldsymbol{X}=(X_1,\ldots,X_J)$ is given by
$$\overline{F}_{\boldsymbol{X}|Z}(\boldsymbol{x}\ | \ z, \boldsymbol{\theta})=\exp\left(-z\sum_{j=1}^J\Lambda_j(x_j)\right).$$
Toward this end, we assume that $Z$ follows a Gamma distribution with mean 1 and variance $\nu$ \citep[]{Hougaard_1995}. Then, the unconditional joint r.f. of $\boldsymbol{X}$ is given by
\begin{align}\label{joint_surv}
\overline{F}_{\boldsymbol{X}}(\boldsymbol{x}\ | \ \boldsymbol{\theta}
)=\left(1+\nu \Delta\right)^{-\frac{1}{\nu}},     
\end{align}   
where $\Delta=\sum_{j=1}^J \Lambda_j(x_j)=\sum_{j=1}^J\left(\frac{x_j}{\eta_j}\right)^\gamma$ and $\boldsymbol{\theta}=(\eta_1,\ldots,\eta_J,\gamma,\nu)$. Then, the sub-density function is given by
\allowdisplaybreaks\bea
g(j,t\ | \ \boldsymbol{\theta})&=&\left(\frac{\gamma}{\eta_j}\right)\left(\frac{t}{\eta_j}\right)^{\gamma-1}\left[1+\nu \sum_{j'=1}^J\left(\frac{t}{\eta_{j'}}\right)^\gamma\right]^{-\frac{1}{\nu}-1}.\nonumber
\eea
Furthermore, the sub-survivor function is given by
\allowdisplaybreaks\bea\label{sub}
\overline{G}(j,t\ | \ \boldsymbol{\theta})&=&\int_t^{\infty} g(j,s) \mathrm{d}s\nonumber\\
&=&\int_t^{\infty} \left(\frac{\gamma}{\eta_j}\right)\left(\frac{s}{\eta_j}\right)^{\gamma-1}\left[1+\nu \sum_{j'=1}^J\left(\frac{s}{\eta_{j'}}\right)^\gamma\right]^{-\frac{1}{\nu}-1} \mathrm{d}s\nonumber\\
&=& \frac{1}{\sum_{j'=1 }^J\left(\frac{\eta_{j}}{\eta_{j'}}\right)^{\gamma}}\left[1+\nu\sum_{j'=1}^J\left(\frac{t}{\eta_{j'}}\right)^{\gamma}\right]^{-\frac{1}{\nu}}.
\eea
A common issue that arises with these frailty models is the problem of identifiability of the model parameters. Using Theorem 2.1 of \citet{ghosh_2024}, it can be easily shown that the model presented above is identifiable from the observed competing risks data on $(C, T)$.
Toward this end, we obtain the r.f. of the lifetime $T$. It is easy to see that the r.f. of $T$ is given by
\begin{align}\label{surv}  
\overline{F}_T(t\ |\  \boldsymbol{\theta})=P(T>t)=P(X_1>t,\ldots, X_J>t)=\overline{F}_{\boldsymbol{X}}(t\boldsymbol{1}|\boldsymbol{\theta})=\left[1+\nu\sum_{j=1}^J\left(\frac{t}{\eta_j}\right)^{\gamma}\right]^{-\frac{1}{\nu}}.
\end{align}
Interestingly, as $\nu\rightarrow 0$, we have
\bea
\overline{F}_{T}({t}\ |\ \boldsymbol{\theta})=\exp\left[-\sum_{j=1}^J\left(\frac{t}{\eta_j}\right)^{\gamma}\right],\label{surv1}
\eea
which is essentially the r.f. of $T$ under the assumption of independence among the potential failure times \citep{Roy_14}. Thus, the independent competing risks model can be obtained as a special case of the shared gamma frailty model presented here. Accordingly, the sub-survivor function is given by 
\begin{align}\label{sub1}
  \overline{G}(j,t\ | \ {\boldsymbol{\theta}})=  \frac{1}{\sum_{j'=1}^J\left(\frac{\eta_{j}}{\eta_{j'}}\right)^\gamma}\exp\left[-\sum_{j'=1}^J\left(\frac{t}{\eta_{j'}}\right)^{\gamma}\right].
\end{align}
The degree of dependence among potential failure times $X_j$s is often measured by the following ratio
$$\rho(t)=\frac{\overline{F}_{\boldsymbol{X}}\left(\boldsymbol{x}|\boldsymbol{\theta}\right)}{\prod_{j=1}^J \overline{F}_{X_j}\left(x_j|\boldsymbol{\theta}_j\right)},$$
where $\overline{F}_{\boldsymbol{X}}(\cdot)$ and $\overline{F}_{X_j}(\cdot)$ are as introduced above  \citep[][pp.~37-38]{crowder2001classical}. If $\rho(t)>1$, there is a positive dependence among the potential failure times $X_j$s, for all $t$. It is easy to show that $\rho(t)$ is an increasing function of $\nu$, and it attains a value of 1 only when $\nu \to 0$. Thus, any positive value of $\nu$ (i.e., $\nu>0$) induces positive dependence among the potential failure times $X_j$s. Fig. \ref{sys_rel} further shows the plot of the r.f. $\overline{F}_T(t|\boldsymbol{\theta})$ for different values of $\nu$.  For convenience, we 
represent the independent competing risks model using $\nu = 0$. It is easy to see that the reliability of a unit improves as $\nu$ increases for a given $t$. 

\begin{figure}[hbt!]
\centering
\includegraphics[width=0.5\textwidth, height=0.25\textheight]{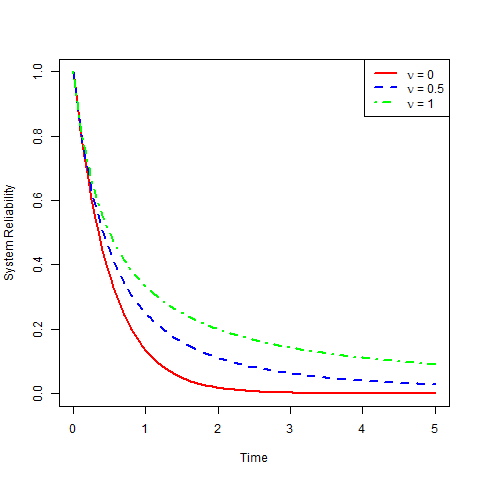}
\caption{Reliability function for different $\nu$ with $\eta_1=\eta_2=\gamma=1$}
\label{sys_rel}
\end{figure}

\section{Development of Sampling Plans}\label{dev}
Suppose $n$ test units, randomly selected from a lot, are put on a life test under a $M$-point PIC-I scheme to decide the acceptability of the lot. In practice, the decision to accept or reject such a lot is often based on the product's reliability at some pre-decided time point mutually agreed upon by the producer and consumer. Thus, a lot is accepted if the reliability of the unit $\overline{F}_T(t)$ at some specified time point $t_0$ is greater than $\pi_0$ and it is rejected if it is lower than $\pi_1$ \citep{Lam_90}.  Obviously, $\pi_0\ge\pi_1$. The appropriate values of $\pi_0$ and $\pi_1$ are decided based on an agreement between the producer and the consumer. It is easy to see that the above decision rule is equivalent to testing the following set of simple hypotheses:
\begin{align}\label{Hyp}
H_0: \overline{F}_T(t_0)=\pi_0 \,\,\,\,\text{vs}\,\,\,\, H_A: \overline{F}_T(t_0)=\pi_1.
\end{align}
In practice, $\overline{F}_T(t_0)$ needs to be estimated. Let $\widehat{\overline{F}}_T(t_0)$ be the MLE of $\overline{F}_{T}(t_0)$. The lot is accepted if $\widehat{\overline{F}}_T(t_0)>\pi_c$ and rejected if $\widehat{\overline{F}}_T(t_0)\le \pi_c$, where $\pi_c$ is an appropriately chosen acceptance limit. In traditional reliability sampling plans, we need to decide $n$ and $\pi_c$.
 Let $\alpha$ and $\beta$ be the producer's risk and consumer's risk, respectively. We have
\begin{align}\label{SP1}
P\left(\widehat{\overline{F}}_{T}(t_0)>\pi_c\ | \ H_0\right)=1-\alpha
\end{align}
and
\begin{align}\label{SP2}
P\left(\widehat{\overline{F}}_{T}(t_0)>\pi_c\ |\ H_A\right)=\beta.
\end{align}

\indent Now, from Theorem \ref{th2}, we get that $\sqrt{n}(\widehat{\boldsymbol{\theta}}-\boldsymbol{\theta})$ follows asymptotically normal distribution with mean $\boldsymbol{0}$ and variance-covariance $I^1(\boldsymbol{\theta}|\boldsymbol{\zeta} )^{-1}$, where $I^1(\boldsymbol{\theta}|\boldsymbol{\zeta})=I(\boldsymbol{\theta}|\boldsymbol{\zeta})/n$. Furthermore, using the Delta method, $\sqrt{n}[\widehat{\overline{F}}_{T}(t_0)-\overline{F}_{T}(t_0)]$ follows asymptotically normal distribution with mean 0 and variance $S^2$, where $S^2=\boldsymbol{C}_T'[I^1(\boldsymbol{\theta}|\boldsymbol{\zeta})]^{-1}\boldsymbol{C}_T$ and $\boldsymbol{C}_T=\nabla_{\boldsymbol{\theta}}\overline{F}_{T}(t_0)$ \citep[][pp. ~619-620]{ME_98}. Hence
$$Z=\frac{\sqrt{n}[\widehat{\overline{F}}_{T}(t_0)-\overline{F}_{T}(t_0)]}{S}$$
is asymptotically normal with a mean of 0 and a standard deviation of 1. Now, as Theorem \ref{theorem1} suggests, we need to compute $\nabla_{\boldsymbol{\theta}}(q_{ij})$ and $\nabla_{\boldsymbol{\theta}}(q_{i})$, which are given in the \ref{computation}.\\
\indent Towards this end, let us denote by $S_0$ and $S_1$ the values of $S$, under $H_0$ and $H_A$, respectively. Then, from (\ref{SP1}) and (\ref{SP2}), we have 
$$\Phi\left(\frac{\sqrt{n}[\pi_c-\pi_0]}{S_0}\right)=\alpha$$
and
$$\Phi\left(\frac{\sqrt{n}[\pi_c-\pi_1]}{S_1}\right)=1-\beta.$$
Solving these two equations, we get
\begin{align}\label{optp}
    \pi_c=\frac{\pi_0 S_1 z_{\beta}-\pi_1 S_0 z_{1-\alpha}}{S_1z_{\beta}-S_0z_{1-\alpha}},
\end{align}
where $z_{\delta}$ is such that $P(Z>z_{\delta})=\delta$. Moreover, using the  expression of $\pi_c$, we get
$n$ as follows:
\begin{equation}\label{optn}
    n=\left[\frac{S_1 z_{\beta}-S_0z_{1-\alpha}}{\pi_0-\pi_1}\right]^2.  
\end{equation}

\indent For convenience, we reproduce the entire process of finding $n$ and $\pi_c$ in Algorithm \ref{alg:PIC-I}.

\begin{algorithm}
\caption{Determination of $n$ and $\pi_c$}\label{alg:PIC-I}
\begin{algorithmic}[1]
\State Fix $\alpha$ and $\beta$, the producer's and consumer's risks, respectively.
\State Fix $t_0$.
\State Fix $\pi_0$ and $\pi_1$, the acceptable and rejectable reliability limits, respectively.
\State Compute $S_0$ and $S_1$, the standard deviations of $\widehat{\overline{F}}(t_0|\boldsymbol{\theta})$ under $H_0$ and $H_A$, respectively. 
\State Compute the acceptance limit $\pi_c$ using (\ref{optp}).
\State Calculate the value of $n$ using (\ref{optn}).
\State The required sample size is taken as $n^*=\lfloor n\rfloor$.

\end{algorithmic}
\end{algorithm}


\section{Optimal PIC-I Schemes}\label{OPT}
\indent In the previous section, we presented the method for obtaining the optimal $n$ and $\pi_c$, for given $\alpha$, $\beta$, $\pi_0$ and $\pi_1$. In practice, however, it is essential to suitably determine inspection times and the proportions of withdrawals to efficiently conduct life tests under PIC-I schemes. In general, the current literature presents two primary approaches to resolve this issue. The first approach deals with achieving the estimation precision of some important reliability characteristics and presents the optimal design parameters considering a suitably constructed design criterion that aligns with this goal. This approach typically ignores the cost constraints that are ubiquitous in the context of life tests. Thus, the second approach proposes a cost-based objective function and determines the optimal design parameters by minimizing such a function. In this context, it is also important to highlight that there are a select number of studies that integrate both these approaches. However, such work is largely absent in the context of PIC-I schemes and, more importantly, in the presence of competing causes of failure. Towards this end, we present the first approach in \S \ref{OPT_WOC} and subsequently deal with the combined approach in \S \ref{OPT_WC}.\\
\subsection{Optimal RASP without Cost Constraints} \label{OPT_WOC}
\indent  Since our main focus is on reliability-based sampling plans in this article, we consider a design criterion that focuses on precise estimation of $\widehat{\overline{F}}_T(t|\boldsymbol{\theta})$. This, in turn, would be helpful in accurately deciding the lot acceptance criteria presented in the previous section. Thus, we consider a design criterion that aims to minimize the variance of $\widehat{\overline{F}}_T(t|\boldsymbol{\theta})$ at a specific time point $t_0$. In particular, our design criterion is given by
\begin{align}\label{des_crit_2}
\phi(\boldsymbol{\zeta})=S^2,
\end{align}
where $S^2$ is the standardized variance, as in the previous section \citep[p.~121]{atk_07}. This design criterion is more commonly known as the $c$-optimality criterion in the literature. It is easy to see that $\phi(\boldsymbol{\zeta})$ depends on the unknown true parameter $\boldsymbol{\theta}$. Thus, in order to compute it, a suitable set of guess values for the parameters is typically used. Note that such guess values are often available in practice from similar life tests performed in the past \citep[Chap.~10]{ME_98}.\\
\indent It is easy to see that $\phi(\zeta)$ is independent of $n$. To compute the optimal $\boldsymbol{\zeta}$, say $\boldsymbol{\zeta}^*$, we minimize $\phi(\boldsymbol{\zeta})$ with respect to $M$, $\boldsymbol{L}_M$ and $\boldsymbol{p}_M$. Toward this end, as is standard in the existing literature, we assume that inspections are equispaced, i.e., for $i=1,\ldots, M$, $L_i-L_{i-1}=h$, where $h$ is the common time interval between two successive inspections. Furthermore, we also assume that $p_1=\ldots=p_{M-1}=p$, where $p$ is the common proportion of withdrawals at each $L_i$. Note that such PIC-I schemes are particularly popular among reliability engineers, mainly because of the straightforwardness they offer in practical applications. Thus, we have three decision variables in this context, resulting in $\boldsymbol{\boldsymbol{\zeta}}=(M, h, p)$. However, as the following theorem shows, the complexity of the optimization problem reduces further due to the following monotonicity properties of $\phi(\boldsymbol{\zeta})$.  


\begin{theorem}\label{theorem5}
    The design criterion $\phi(\boldsymbol{\boldsymbol{\zeta}})$, given by (\ref{des_crit_2}), satisfies the following monotonicity properties:
    \begin{enumerate}[(I)]
        \item For fixed $(h,p)$, $\phi(\boldsymbol{\boldsymbol{\zeta}})$ is decreasing in $M$.
        \item For fixed $(M, h)$, $\phi(\boldsymbol{\boldsymbol{\zeta}})$ is increasing in $p$.       
    \end{enumerate}
\end{theorem}
\begin{proof}
    The proof of the theorem is given in \ref{theoremapp}.
\end{proof}

Theorem \ref{theorem5} suggests that $\phi(\boldsymbol{\boldsymbol{\zeta}})$ attains its optimum value for the decision variables $M$ and $p$ at their respective limits. Thus, we must fix them \textit{a priori} and determine the optimal value of $h$, say $h^*$, by minimizing $\phi(\boldsymbol{\zeta})$. Subsequently, we determine $\pi_c$ and $n$ using (\ref{optp}) and (\ref{optn}), respectively. For convenience, we summarize this entire process in Algorithm \ref{alg2:PIC-I}.


\begin{algorithm}
\caption{Determination of Optimal Decision Variables without Cost Constraints} \label{alg2:PIC-I}
\begin{algorithmic}[1]
\State Fix $\alpha$ and $\beta$, the producer's and consumer's risks, respectively.
\State Fix $t_0$.
\State Fix $\pi_0$ and $\pi_1$, the acceptable and rejectable reliability limits, respectively.
\State Fix $M$ and $p$, say $M_0$ and $p_0$, respectively.
\State Obtain $h^*$, the optimal value of $h$, by minimizing $S^2$. \Comment{The expression for $S^2$ is provided in Section \ref{dev}.}
\State Calculate the value of $n$ from (\ref{optn}) using $h^*$.
\State The required sample size is taken as $n^*=\lfloor n\rfloor$.
\State Compute the acceptance limit $\pi_c$ from (\ref{optp}) using  $\boldsymbol{\boldsymbol{\zeta}}^*=(n^*,M_0,h^*,p_0)$.
\end{algorithmic}
\end{algorithm}
\indent Note that the optimization problem involved in Step 5 in Algorithm \ref{alg2:PIC-I} is rather simple and can be easily solved using any standard optimization module (e.g., \textit{optim()} function in \textit{R}). 

\subsection{Optimal RASP with cost constraints} \label{OPT_WC}
\indent Reliability experiments are often conducted under strict budgetary limitations; the approach presented in the previous subsection overlooks these budgetary constraints. To resolve this issue, we now incorporate the budgetary constraints into the optimization problem presented above and provide an enhanced solution that aligns with real-world budgetary constraints.\\
\indent Let $TC(\boldsymbol{\boldsymbol{\zeta}})$ denote the total cost of running a life test under a PIC-I scheme. Note that the total cost depends on the following costs as follows:
\begin{enumerate}[(i)]
    \item Sample Cost: This is the cost of $n$ test units that are subjected to a life test under a PIC-I scheme. Let $C_S$ be the cost of each test unit. Then, the total cost of the sample is given by $nC_S$.
    \item Operation Cost: This involves the cost of running a life test, including the cost of utilities and the salary of the operators. Let $\tau$ denote the duration of the life test. Furthermore, let $C_\tau$ be the operating cost per unit of time. Then the resulting total operating cost is $C_\tau E[\tau]$, where $E[\tau]$ is the expected duration of the PIC-I scheme and is given by
    \begin{align*}
    E[\tau]=\sum_{m=1}^M L_m \mathcal{P}_m, 
    \end{align*}
    where $\mathcal{P}_m$ is the probability that the life test terminates at $L_m$ and is given by
\begin{align*}
   \mathcal{P}_m= P_m-P_{m-1}, 
\end{align*}
with
\begin{align*}
P_m=&\left[\sum_{l_1=1}^m \prod_{l_2=0}^{l_1-1}(1-p_{l_2})(1-q_{l_2})\{q_{l_1}+(1-p_{l_1})q_{l_1}\}\right]^n,
\end{align*}
for $m=1,\ldots,M-1$. Note that $P_0=0$ and $P_M=1$ \citep{sonal_17}.

\item Failure Cost: This refers to the expenses incurred from the failure of a test unit during a life test experiment, encompassing the costs associated with identifying the failure mode, performing a rework, scrapping the defective unit, and any additional procedures necessary to understand the failure. Let $C_D$ be the cost associated with a failed test unit. So, the resulting total cost is given by $C_DE[D]$, where $E[D]$ is given by (\ref{dsum}).  

 
\item Inspection Cost: This cost involves the cost of conducting intermittent inspections during the course of the life test under a PIC-I scheme. Let $I$ denote the number of inspections completed in a $M$-point PIC-I scheme. Furthermore, let $C_I$ be the cost per inspection in such a life test. Then, the total inspection cost is given by $C_I E[I]$, where $E[I]$ is expected number of inspections and is given by
 \begin{align*}
    E[I]=\sum_{m=1}^M m \mathcal{P}_m.
\end{align*}
\end{enumerate}
Thus, the total cost corresponding to a life test under a PIC-I scheme is given by
\begin{align*}
    TC(\boldsymbol{\boldsymbol{\zeta}})=nC_S+C_\tau E[\tau]+C_DE[D]+C_IE[I].
\end{align*}
Let $C_B$ be the available budget for conducting the life test. Then, for given $(\alpha,\beta)$, the decision problem to find the optimal $\boldsymbol{\zeta}=(n,M,h,p)$ is given by
\begin{align}\label{min}
    \text{minimize }\phi(\boldsymbol{\boldsymbol{\zeta}})
\end{align}
subject to
\begin{align}\label{cost}
    TC(\boldsymbol{\zeta})\leq C_B
\end{align}
\begin{align}\label{n1}
  n=\left[\frac{S_1 z_{\beta}-S_0z_{1-\alpha}}{\pi_0-\pi_1}\right]^2 
\end{align}
\begin{align*}
     M\geq s,~n\in\mathbb{N},~M\in \mathbb{N},~h>0 \text{ and } ~ 0\leq p< 1.
\end{align*}

It is obvious that (\ref{cost}) represents the cost constraint, as discussed above. Note that the equality constraint in ($\ref{n1}$) is derived in Section \ref{dev}. Furthermore, as the Regularity Condition VII(b) suggests, $M$ must be greater than $s$, the number of parameters in the model. Toward this end, we note that $s=4$ for the dependent model given by (\ref{sub}). Similarly, for the independent model in (\ref{sub1}), $s=3$. \\
\indent From (\ref{n1}), it is easy to see that $n$ depends on $M$, $h$ and $p$. We reduce the complexity of the decision problem by substituting (\ref{n1}) both in (\ref{min}) and (\ref{cost}). Furthermore, as is standard in the existing literature \citep{huang_08}, we pre-fix the uniform withdrawal proportion $p$. Thus, the optimization problem involves only the decision variables $M$ and $h$. Furthermore, as Theorem \ref{theorem5} suggests, $\phi(\boldsymbol{\zeta})$ is decreasing in $M$, while $TC(\boldsymbol{\zeta})$ is increasing in $M$. Toward this end, we propose Algorithm \ref{alg3:PIC-I} to compute the optimal decision variables.
\begin{algorithm}
\caption{Determination of Optimal Decision Variables with Cost Constraints} \label{alg3:PIC-I}
\begin{algorithmic}[1]
\State Fix $\alpha$ and $\beta$, the producer's and consumer's risks, respectively.
\State Fix $t_0$.
\State Fix $\pi_0$ and $\pi_1$, the acceptable and rejectable reliability limits, respectively. 
\State Fix the uniform withdrawal proportion $p$, say $p_0$.
\State Fix an upper bound for $M$, say $M_0$.
\State $M \gets s$. \Comment{$s$ is the the lower bound for $M$}
\While{$M<M_0$} 
\State Obtain the solution $h^*(M)$ that minimizes $\phi(\boldsymbol{\zeta})$ subject to cost constraint given by (\ref{cost}).
\State Compute $n$ using (\ref{n1}).
 \State Set $n^*(M)=\lfloor n\rfloor$. \Comment{$n^*(M)$ is the optimal value of $n$ for a given $M$.}
 \State Compute $\phi(\boldsymbol{\zeta})$, for $\boldsymbol{\zeta}=(n^*(M), M, h^*(M),p_0)$. 
\State $M \gets M+1$.
\EndWhile
\State Obtain $M^*$ such that \Comment{ $M^*$ is the optimal value of $M$.}
$$\phi(n^*(M^*),M^*,h^*(M^*),p_0)=\min_{s\leq M\leq M_0}\phi(n^*(M),M^*,h^*(M),p_0).$$ 
\State Set $n^* \equiv n^*(M^*)$ and $h^* \equiv h^*(M^*)$. 
\State Compute the acceptance limit $\pi_c$ from (\ref{optp}) using  $\boldsymbol{\boldsymbol{\zeta}}^*=(n^*,M^*,h^*,p_0)$.
\end{algorithmic}
\end{algorithm}

\section{Numerical Experiment} \label{num_exp}
\indent In this section, we conduct a detailed numerical study to assess how the
degree of dependence among the potential failure times impacts the resulting RASPs. We consider the battery failure data set introduced in Section \ref{INTRO} after recalibrating it in units of thousands of ampere-hours to determine the necessary background details.\\
\indent The battery failure data set includes four causes of failure and provides the number of failures attributed to each cause at each inspection point, along with the corresponding number of withdrawals. The limited number of observed failures for each cause leads us to group Causes "1" and "3" together as the first cause of failure. Similarly, Causes "2" and "4" are together considered as the second cause of failure. Therefore, we have $J=2$ for the subsequent analysis. We also use the independent competing risk model presented in (\ref{sub1}) and obtain the MLEs $\widehat{\eta}_1=1.291$, $\widehat{\eta}_2=1.339$ and $\widehat{\gamma}=1.644$ by maximizing the log-likelihood function $l(\boldsymbol{\theta}\ | \ \mathcal{D}_M,\boldsymbol{\boldsymbol{\zeta}})$ in (\ref{log_like_eqn}).\\
\indent As Algorithms \ref{alg:PIC-I},  \ref{alg2:PIC-I} and \ref{alg3:PIC-I} suggest, we further need to compute $S_0$ and $S_1$, the values of $S$ under $H_0$ and $H_A$ as specified in (\ref{Hyp}). However, this would require the specification of the unknown model parameters. Toward this end, we adopt the approach of \citet{Wu_2020}. We note that the reliability functions $\overline{F}_T(t_0|\boldsymbol{\theta})$ in (\ref{surv}) is an increasing function of $\eta_j$ for each $j=1,\ldots,J$. Thus, the hypotheses in (\ref{Hyp}) can be equivalently stated as follows:
$$H_0: \boldsymbol{\eta}=\boldsymbol{\eta}_0\,\,\,\,\text{vs}\,\,\,\, H_A: \boldsymbol{\eta}=\boldsymbol{\eta}_1,$$
where $\boldsymbol{\eta}_0=(\eta_{01},\ldots,\eta_{0J})$ and $\boldsymbol{\eta}_1=(\eta_{11},\ldots,\eta_{1J})$. A lot is accepted when $\boldsymbol{\eta}$ is greater than $\boldsymbol{\eta}_0$ and it is rejected when $\boldsymbol{\eta}$ is less than $\boldsymbol{\eta}_1$. Obviously, $\eta_{0j} \ge \eta_{1j}$, for $j=1,\ldots,J$. For the subsequent analysis, we consider $\boldsymbol{\eta}_0=(\eta_{01},\eta_{02})=(1.291,1.339)$, which are the MLEs of $\eta_1$ and $\eta_2$. To decide $\boldsymbol{\eta}_1$, we further define the discrimination ratio $d_j$ as $d_j=\eta_{0j}/\eta_{1j}$ and set it to a suitable set of values as detailed in the subsequent sub-sections. Throughout the analysis, we have used $\gamma=1.644$, which is the MLE of $\gamma$. Furthermore, we consider $(\alpha,\beta)=(0.05,0.1)$. Also, the specific time point of interest $t_0$ is fixed at 0.5.\\    
\indent The rest of this section is organized as follows. In \S \ref{NUM_OPT}, we provide optimal sampling plans following the methodology presented in Section \ref{dev}. Subsequently, we consider the methodology presented in Section \ref{OPT} and provide the $c$-optimal design parameters along with the optimal sample size and acceptance limit in \S \ref{NUM_C_OPT} without the budgetary constraints and in \ref{NUM_C_OPT} with the budgetary constraints.
\subsection{Optimal Sampling Plans} \label{NUM_OPT}
\indent As evident from Section \ref{dev}, we must specify $M$, $\boldsymbol{L}_M$ and $\boldsymbol{p}_M$ first to obtain the optimal sample size $n$ and the acceptance limit $\pi_c$. For convenience, we consider the equispaced PIC-I schemes with uniform withdrawal proportions introduced in Section \ref{OPT}. In particular, we consider three options for the number of inspections $M$ (i.e., 4, 6, and 8). In each case, we further select three choices for both the common inspection time interval $h$ (i.e., 0.2, 0.3, and 0.4) and the common withdrawal proportion $p$ (i.e., 0.0, 0.2, 0.3). Moreover, we consider the same discrimination ratio for both causes and choose two possible values for the common discrimination ratio $d$ (i.e., 1.5 and 1.8). We then obtain the optimal values of $\pi_c$ and $n$ both for the independent and dependent competing risks model presented in Section \ref{Frail_Mod}. For the dependent competing risks model, we further consider two different levels of dependence considering two different choices for $\nu$ (i.e., 0.5 and 1). Table \ref{Table 1} provides the resulting optimal values of $\pi_c$ and $n$ in each case. For convenience, we represent the independent competing risks model using $\nu=0$ in Table \ref{Table 1}.\\
\begin{table}[hbt!]
   \centering
    \caption{Optimal Sampling Plans}
    \begin{tabular}{|c|cc|cc|cc|cc||cc|cc|cc|}
    \hline
 \multirow{3}{*}{} &\multirow{3}{*}{$\nu$}  &\multirow{3}{*}{$h$}&\multicolumn{6}{c||}{$d=1.5$}&\multicolumn{6}{c|}{$d=1.8$}\\
  \cline{4-15}
  & && \multicolumn{2}{c|}{$M=4$} & \multicolumn{2}{c|}{$M=6$}& \multicolumn{2}{c||}{$M=8$}&\multicolumn{2}{c|}{$M=4$} & \multicolumn{2}{c|}{$M=6$}& \multicolumn{2}{c|}{$M=8$}\\
    \cline{4-15}

  \multirow{9}{*}{\rotatebox{90}{$p=0$}}  
&  &   &$\pi_c$&$n^*$&$\pi_c$&$n^*$&$\pi_c$&$n^*$&$\pi_c$&$n^*$&$\pi_c$&$n^*$&$\pi_c$&$n^*$\\
  \cline{1-15}
&   &    0.2& 0.547& 32& 0.546& 31& 0.547& 31&0.482 &13&0.481& 13&0.481& 13\\
 & 0&    0.3& 0.546& 33& 0.547& 33& 0.547& 33&0.481 &14& 0.482& 13&0.482& 13\\
&& 0.4& 0.548 &36&  0.548& 36& 0.548 &36&0.483&15&0.484&15&0.484 &15\\
  \cline{2-15}
 &    &   0.2&  0.593 & 57& 0.594 &51&0.594 &49& 0.544 &  25 &0.546 &  23&0.547 &  22\\
  &0.5&    0.3& 0.594 &53& 0.594 &50& 0.594 &49&0.547 &  24& 0.547 &  22 &0.547 &  22\\
  &&0.4&0.593& 53&0.594& 53&0.593 &53& 0.546 &  24& 0.545  & 23 & 0.545 &  23\\
    \cline{2-15}
& &   0.2&0.627& 76&0.628& 70& 0.629 &67&0.587& 34&0.589& 32&0.590& 30\\
 &1&    0.3&0.629 &72&0.628& 68&0.628 &67& 0.590& 33& 0.590& 31& 0.590& 30\\
&& 0.4&  0.628& 71&0.628 &70&0.628 &70& 0.589 &32&0.588& 32&0.588& 32\\
   
       \hline
         \hline
 \multirow{9}{*}{\rotatebox{90}{$p=0.2$}}&   &    0.2&  0.548 &38& 0.547& 37&  0.547& 37& 0.484 &16&0.484 &15& 0.484& 15\\
&  0&    0.3& 0.547 &36&0.547& 36&0.547 &36& 0.482 &15&0.482 &15& 0.482& 15\\
&& 0.4&0.547& 37&    0.547 &37&   0.547& 37& 0.482 &15& 0.483& 15& 0.483& 15\\
  \cline{2-15}
  &   &   0.2&0.593& 71& 0.594& 66&  0.595& 63&0.544 &  31& 0.546 &  29&0.547 &  28\\
&  0.5&    0.3& 0.595& 60&0.595& 56&0.595& 55& 0.548 &  27 &0.549 &  25 &0.549&   25\\
&&0.4&  0.594& 56& 0.594& 55&  0.594 &54& 0.547 &  25& 0.546 &  24 &0.546  & 24\\
     \cline{2-15}
            &&   0.2& 0.627& 93& 0.628 &88&0.628 &86&0.587 &42& 0.589 &40& 0.589& 39\\
&1&    0.3&  0.629& 82&0.629& 77& 0.629 &75&0.590& 38& 0.591 &35& 0.591& 35\\
&& 0.4&  0.628& 75&   0.628 &73&0.628 &73&0.589 &34&0.589& 33&0.589 &33\\
     \hline
         \hline
 \multirow{9}{*}{\rotatebox{90}{$p=0.3$}}&   &    0.2&  0.548& 44&  0.548 &42&  0.548&42& 0.486 &18&0.486& 17& 0.486& 17\\
&  0&    0.3& 0.547& 37&0.547& 37&0.547 &37& 0.483 &15&0.483 &15& 0.483& 15\\
&& 0.4&0.547& 38&    0.547 &38&   0.547& 38& 0.482 &16& 0.483& 15& 0.483& 15\\
  \cline{2-15}
  &   &   0.2&0.593& 82& 0.594& 77&   0.595& 74&0.544& 36& 0.546& 34& 0.547& 33 \\
&  0.5&    0.3& 0.595& 66&0.595& 61&0.595& 60& 0.549& 29&0.549& 27&0.550& 27 \\
&&0.4&  0.594&57& 0.594& 56&  0.594 &56& 0.548 &  26& 0.547 &  25 &0.547  & 25\\
     \cline{2-15}
            &&   0.2&  0.627& 107& 0.628& 102&0.628& 100& 0.587& 48& 0.588& 46& 0.589& 45\\
&1&    0.3&  0.629& 89&0.629& 83&  0.629 &82&0.591& 41&0.591 &38& 0.591 &38\\
&& 0.4&  0.628& 77&  0.628 &75&0.628 &75& 0.590& 35& 0.590& 34&0.589 &34\\
     \hline
       \end{tabular}
    \label{Table 1}
\end{table}
\indent It is interesting to note that as $\nu$ increases, both $\pi_c$ and $n$ increase irrespective to the choices of $M$, $h$, $p$, and $d$. Furthermore, the resulting optimal $\pi_c$ and $n$ are higher for the dependent competing risks model compared to the independent competing risks model. This is perhaps due to the fact that the r.f. increases with $\nu$. For given $\nu$, $h$, $M$, the optimal $\pi_c$ is rather insensitive to the choice of withdrawal proportion $p$. However, the optimal $n$ increases with $p$ both for the dependent and independent competing risks model. Furthermore, the effect of $M$ on $\pi_c$ and $n$ seems to be insignificant, particularly for higher $d$. As expected, both $\pi_c$ and $n$ decrease as $d$ increases.\\
\indent As is standard in the existing literature, we now present the Operating Characteristic (OC) curves in Figure \ref{Oc curve} for both the independent and dependent competing risks models for $M$=6, $h=0.3$, and $p=0.2$. Note that the OC curve demonstrates the likelihood of accepting a lot against the proportion of defective items in the lot. As expected, the lot acceptance probability decreases as the proportion of defective items goes up. Moreover, for a fixed proportion of defective items, the acceptance probability is lower for the dependent competing risks model for both $d=1.5$ and $d=1.8$. This reconfirms the importance of incorporating the dependence among the potential risks while modeling competing risks data.
We also observe that the OC curves become flatter as $d$ increases, suggesting that the lot acceptance probability increases with $d$ for a given proportion of defective items.



\begin{figure}[hbt!]
\begin{subfigure}{0.45\textwidth}
\centering
\includegraphics[width=\textwidth, height=0.25\textheight]{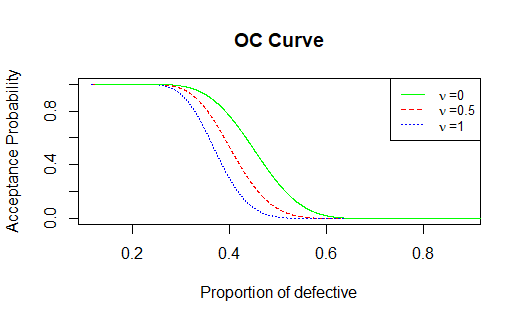}
\caption{$d=1.5$}
\label{Oc1}
\end{subfigure}
\begin{subfigure}{0.45\textwidth}
\centering
\includegraphics[width=\textwidth, height=0.25\textheight]{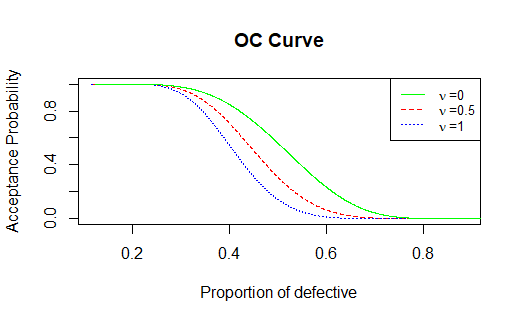}
\caption{$d=1.8$}
\label{Oc2}
\end{subfigure}
    \caption{{OC curve when $h=0.3$ $M=6$ and $p=0.2$}}
    \label{Oc curve}
\end{figure}
\subsection{Optimal RASP without Budgetary Constraint} \label{NUM_C_OPT}
\indent We now consider the determination $h$, $n$ and $\pi_c$ using the $c$-optimality design criterion presented in Section \ref{OPT}. As Algorithm \ref{alg2:PIC-I} suggests, we first fix $M$ and $p$ at the three possible values presented above and obtain the optimal $h$ by minimizing $\phi(\boldsymbol{\zeta})$. 
Table \ref{Table 2} reports the resulting optimal interval duration.  Note that the optimal inspection interval $h$ is independent of the discrimination ratio $d$. However, as (\ref{optp}) and (\ref{optn}) suggest, both $\pi_c$ and $n$ depend on the discrimination ratio $d$. Thus, Table \ref{Table 2} reports $\pi_c$ and $n$ for both $d=1.5$ and $1.8$. Furthermore, Table \ref{Table 2} also provides the optimal value of $S^2$.\\   
\begin{table}[hbt!]
    \centering
\small    \caption{Optimal Plans without Budgetary Constraints }
    \begin{tabular}{|c|c|cccc|cccc|cccc|}
    \hline
 \multirow{2}{*}{$p_0$}& \multirow{2}{*}{$\nu$} & \multicolumn{4}{|c|}{$M_0=4$} & \multicolumn{4}{c|}{$M_0=6$}& \multicolumn{4}{c|}{$M_0=8$}\\
    \cline{3-14}
& &\multirow{2}{*}{$h^*$}&$10\times$& \multirow{2}{*}{$n^*$}&\multirow{2}{*}{$\pi_c$}&\multirow{2}{*}{$h^*$}&$10\times$&\multirow{2}{*}{$n^*$}&\multirow{2}{*}{$\pi_c$}&\multirow{2}{*}{$h^*$}&$10\times $&\multirow{2}{*}{$n^*$}&\multirow{2}{*}{$\pi_c$}\\
&&&$\phi(\boldsymbol{\boldsymbol{\zeta}}^*)$&&&&$\phi(\boldsymbol{\boldsymbol{\zeta}}^*)$&&&&$\phi(\boldsymbol{\boldsymbol{\zeta}}^*)$&&\\
     \hline
 \multirow{3}{*}{0}   &  \multirow{1}{*}{0}&0.197&  1.649 &32&  0.547 &0.143&  1.601& 31& 0.547&  0.112& 1.579 &30&  0.547\\
 &    \multirow{1}{*}{0.5}&0.332&  1.789& 53& 0.594&0.265&  1.695 &50&  0.594&0.226&  1.650& 49&  0.594\\
     &\multirow{1}{*}{1}& 0.348&  1.767& 71&  0.628&0.284& 1.683 &68&  0.629&0.246&  1.643 &66&  0.629\\
    \hline
 \multirow{3}{*}{0.2}      &      \multirow{1}{*}{0}& 0.272&  1.857 &36&  0.547& 0.282& 1.854 &36 & 0.547& 0.284& 1.850& 36&  0.547\\
  &   \multirow{1}{*}{0.5}&0.384&  1.896& 56&  0.594&0.351& 1.852& 55&  0.594&0.338& 1.839& 54& 0.594\\
    & \multirow{1}{*}{1}&0.40&  1.857& 75&  0.620& 0.371& 1.819& 73&  0.628&0.36&  1.806& 73& 0.628\\
    \hline
 \multirow{3}{*}{0.3}      &      \multirow{1}{*}{0}&0.317&  1.931& 37& 0.547&0.325&  1.927& 37&  0.547& 0.325& 1.926 &37&  0.547\\
  &   \multirow{1}{*}{0.5}& 0.407&  1.945& 57&  0.594& 0.385&  1.910 &56&  0.594& 0.378 & 1.901& 56&  0.594\\
    & \multirow{1}{*}{1}& 0.421 & 1.892& 76& 0.628&0.402 & 1.866 &75&  0.628& 0.396&  1.865& 75&  0.628\\
     \hline
    \end{tabular}
    \label{Table 2}
\end{table}
\indent It is evident that the optimal inspection interval $h$ gets larger as $\nu$ increases, irrespective of the choices of $p$ and $M$. Thus, as the level of dependence among the potential causes of failure goes up, the successive inspections will be more infrequent. Furthermore, as the withdrawal proportion $p$ increases, the optimal interval $h$ becomes longer irrespective of the choices $\nu$ and $M$. Thus, the successive inspections in a PIC-I scheme are expected to be less frequent in comparison to a traditional interval censoring scheme. Furthermore, as the number of inspections $M$ increases, the inspection intervals are slightly wider. As observed in the first case, the sample size $n$ increases as the level of dependence gets stronger. This is also true for $\pi_c$. Furthermore, as expected, the sample size $n$ decreases significantly as the discrimination ratio $d$ increases. We also note that the value of the design criterion increases as the withdrawal proportion $p$ increases, whereas it goes down as $M$ increases. Note that this is consistent with Theorem \ref{theorem5}.   
\subsection{Optimal RASP with Budgetary Constraint}\label{NUM_COST}
\indent In this subsection, we consider the budgetary constraint and obtain the optimal decision variables $n$, $h$, and $\pi_c$ following Algorithm \ref{alg3:PIC-I}. For the purpose of illustration, we assume $C_S=0.1$, $C_{\tau}=5$, $C_D=0.025$, and $C_I=10$, where each of these cost components is as in \S \ref{OPT_WC}. Furthermore, we consider four possible values of the available budget $C_B$, i.e., 55, 65, 85, and 95, primarily to check the sensitivity of the design parameters with respect to available budgetary allocation. Table \ref{NUM_COST_TAB_1} presents the optimal values of $n$, $M$, $h$, $\pi_c$ under various scenarios for the common discrimination ratio $d=1.5$. Table \ref{NUM_COST_TAB_1} further reports the expected number of failures $E[D^{*}]$, expected duration of life test experiment $E[\tau^*]$, expected number of inspections $E[I^*]$ and the optimal value of the design criterion $\phi(\zeta^*)$. The corresponding results for the common discrimination ratio $d=1.8$ are
given in Table \ref{NUM_COST_TAB_2}. \\
\indent We observe that the required sample size and acceptance limit increase as the level of dependence gets stronger for both $d=1.5$ and $1.8$. This is true irrespective of the size of the available budget and the withdrawal proportion $p$. Also, the inspections are expected to get infrequent as the level of dependence goes up in each case. However, the optimal number of inspections is typically insensitive to the level of dependence at a fixed budget and withdrawal proportion. However, for a given level of dependence, as the available budget increases, the number of inspections $M$ goes up marginally. This is true for every withdrawal proportion. Interestingly, the sample size is rather insensitive to the available budget for a given level of dependence. However, the inspection interval is expected to be shorter as the available budget increases. As expected, the sample size requirement significantly drops as the discrimination ratio goes up. Furthermore, the acceptance limit also goes down with the increase in the discrimination ratio.\\
\indent It is also interesting to note that the expected number of failures increases as the level of dependence becomes stronger. This holds for every available budget and withdrawal proportion. Furthermore, the experiment duration is expected to get longer as the level of dependence increases. However, the expected number of inspections is rather insensitive to the level of dependence.

\begin{table}[hbt!]
 \centering
     \caption{Optimal Plans with Budgetary Constraints for $d=1.5$}
\begin{tabular}{|c|cccccccc|}
\hline
  $p_0$&$C_B$& $\nu$&$(n^*,M^*,h^*)$&$\pi_c$&$E[D]$&$E[\tau^*]$&$E[I^*]$& $10\times\phi(\boldsymbol{\boldsymbol{\zeta}}^*)$\\
  \hline
  \multirow{12}{*}{0}& &0    & (32,  4,  0.196)& 0.547& 18.394&  0.784&  4.000&  1.650\\
&55&0.5&(53,  4,  0.332)&   0.594& 39.973&  1.328&  4.000&  1.789\\
&&1& (71,  4,  0.330)&  0.628& 47.443&  1.320&  4.000& 1.770\\
\cline{2-9}
& &0    & (31,  5,  0.165)&  0.547& 18.769&  0.825&  5.000& 1.619\\
&65&0.5&(51,  5,  0.293)&  0.594& 40.417&  1.465&  5.000&  1.732\\
&&1&(71,  4,  0.348)&  0.628& 48.797&  1.392&  4.000&  1.766\\
\cline{2-9}

&&0&(30, 7, 0.125)& 0.547& 19.230&  0.875&  7.000&  1.588\\
&85&0.5&(50,  6,  0.265)&0.594& 41.079&  1.590&  6.000&  1.695\\
&&1&(68, 6, 0.284)& 0.629& 51.270&  1.704&  6.000&  1.683\\
\cline{2-9}
&&0&(30,  7,  0.125)& 0.547& 19.230&  0.875&  7.000&  1.588\\
&95&0.5&(49,  7,  0.243)&  0.594& 41.332&  1.701&  7.000&  1.669\\
&&1&(67,  7,  0.263)&  0.629& 52.045&  1.841&  7.000&  1.661\\
\hline
  \hline
  \multirow{12}{*}{0.2}&&0    & (36,  4,  0.272)&  0.547& 20.095&  1.088&  4.000&  1.857  \\
& 55&0.5&(56,  4,  0.384)& 0.594& 35.120&  1.536&  4.000&  1.896\\
&&1&(85,  4,  0.277)& 0.629& 38.166& 1.108&  4.000&  2.087\\
\cline{2-9}
&&0    & (36,  5,  0.274) &  0.547& 21.833 & 1.362 & 4.970 & 1.857\\
& 65&0.5&(56,  4,  0.384)& 0.594&  35.120&  1.536&  4.000&  1.896\\
&&1&(75,  4,  0.400)&  0.628& 43.207&  1.60&  4.000& 1.853\\
\cline{2-9}
&&0&(36,  8,  0.283)&  0.547& 23.350& 1.748&  6.177&  0.1850\\
&85&0.5&(55, 6, 0.350)&  0.594& 35.138&  2.075&  5.930&  1.852\\
&&1& (73,  6,  0.371)& 0.628& 43.406&  2.224&  5.996&  1.819\\
\cline{2-9}
&&0&(36,  8,  0.283)&  0.547&  23.350& 1.748&  6.177& 1.850\\
&95&0.5&(54,  8,  0.337)&  0.594&  34.526&  2.465&  7.314&  1.839\\
&&1&(73,  7,  0.364)&  0.628& 43.638&  2.536&  6.966& 1.810\\
\hline
\hline
  \multirow{12}{*}{0.3}&&0    & (37,  4,  0.316)&  0.547& 19.940&  1.259&  3.985&  1.931\\
& 55&0.5&(57,  4,  0.406) & 0.594& 32.571&  1.623&  3.997&  1.942\\
&&1& (111,   4,   0.178) & 0.627&   27.896&   0.712&   4.000&   2.847\\
\cline{2-9}
& &0    & (37,  6,  0.324) &  0.547 &21.309 & 1.625 & 5.014 & 1.927\\
&65&0.5&(57,  4,  0.406) & 0.594& 32.571&  1.623&  3.997&  1.942\\
&&1&(76,  4,  0.420)&  0.628& 39.972&  1.680&  4.000&  1.897\\
\cline{2-9}
& &0    & (37,  6,  0.324) &  0.547& 21.309 & 1.625 & 5.014 & 1.927\\
&85&0.5&(56, 9, 0.376)& 0.594&  32.014&  2.417&  6.429&  1.904\\
&&1& (75, 6, 0.401)& 0.628&  40.001&  2.372&  5.916&  1.869\\
\cline{2-9}
&&0&(36,  8,  0.283)&  0.547&  23.350& 1.748&  6.177&  1.849\\
&95&0.5&(56,  9,  0.376)&  0.594&  32.014&  2.417&  6.429&  1.904\\
&&1&(75,  7,  0.397)&  0.6280& 40.037&  2.650&  6.675&  1.865\\
\hline
\end{tabular}
\label{NUM_COST_TAB_1}
\end{table}
\begin{table}[hbt!]

 \centering
     \caption{Optimal Plans with Budgetary Constraints for $d=1.8$}
\begin{tabular}{|c|cccccccc|}
\hline
  $p_0$&$C_B$& $\nu$&$(n^*,M^*,h^*)$&$\pi$&$E[D^*]$&$E[\tau^*]$&$E[I^*]$& $\phi(\boldsymbol{\boldsymbol{\zeta}}^*)\times10$\\
  \hline
  \multirow{12}{*}{0}& &0    &  (13,  4,  0.196)&   0.482&  7.473&  0.784&  4.000& 1.650\\
&55&0.5&(23,  4,  0.332)&   0.547& 17.347&  1.328&  4.000&  1.789\\
&&1& (32,  4,  0.3480) &  0.590&21.993&  1.392&  4.000& 1.766\\
\cline{2-9}
& &0    & (13,  5,  0.165)& 0.482&  7.871&  0.825&  5.000&  1.619\\
&65&0.5&(23,  5, 0.293)&  0.5470& 18.2270&  1.4650&  5.0000 & 1.732\\
&&1& (32,  5,  0.311)&   0.590& 23.200&  1.555&  5.000& 1.716\\
\cline{2-9}
&&0& (12,  7,  0.125)&0.482&  7.692&  0.875&  6.999&  1.588\\
&85&0.5&(22,  7, 0.243)&  0.547& 18.557&  1.699&  6.994&  1.669\\
&&1&(31, 7, 0.263)& 0.590& 24.080&  1.841&  7.000&  1.661\\
\cline{2-9}
&&0&(12, 8,  0.112)&  0.482&   7.864&  0.896&  7.999&  1.571\\
&95&0.5&(21, 8, 0.225)& 0.547& 18.069&  1.796&  7.982&  1.649\\
&&1&(30, 8,  0.246)& 0.590& 23.857&  1.968&  8.000&  1.643\\
\hline
  \hline
  \multirow{12}{*}{0.2}&&0    & (15, 4, 0.272)& 0.482&  8.373&  1.079&  3.968&  1.857 \\
& 55&0.5&(25, 4,  0.384)& 0.547& 15.679&  1.531&  3.986&  1.896\\
&&1&(34, 4,  0.400)& 0.589& 19.587&  1.600&  3.999&  1.833\\
\cline{2-9}
 & &0    & (15, 7,  0.283)&  0.482&  9.691& 1.493&  5.277&  1.851\\
&65&0.5&(24, 5, 0.363)&0.548& 15.263&  1.769&  4.874&  1.867\\
&&1& (34,  5,  0.382)&0.590& 20.005&  1.904&  4.984&  1.819\\
\cline{2-9}

& & 0&(15, 7,  0.283)&  0.482&  9.691& 1.493&  5.277&  1.851 \\
&85 &0.5&  (24, 10,  0.332)&  0.548& 15.327&  2.267&  6.829&  1.834\\
&&1& (33, 7,  0.364)&  0.590& 19.727&  2.439&  6.701&  1.811\\
\cline{2-9}
& & 0&(15, 7,  0.283)&  0.482&  9.691& 1.493&  5.277&  1.851 \\
&95 &0.5&  (24, 10,  0.332)&  0.548& 15.327&  2.267&  6.829&  1.834\\
&&1&(33,  8,  0.359)&  0.590& 19.768&  2.641 & 7.356 & 1.806\\
\hline
\hline
  \multirow{12}{*}{0.3}&&0    & (15,  6,  0.324)& 0.483&  8.639&  1.404 & 4.334&  1.927\\
& 55&0.5&(25, 4, 0.406)&  0.547& 14.286&  1.592&  3.920&  1.942\\
&&1&  (35, 4,  0.420)& 0.589& 18.408&  1.675&  3.988&  1.897\\
\cline{2-9}
&&0    &(15,  6,  0.324)& 0.483&  8.639&  1.404 & 4.334&  1.927  \\
& 65&0.5&(25, 6, 0.384)&0.548& 14.345&  1.932&  5.031&  1.912 \\
&&1&(34, 5,  0.408)&  0.589& 18.074&  1.990&  4.876&  1.878\\
\cline{2-9}

& & 0&(15,  6,  0.324)& 0.483&  8.639&  1.404 & 4.334&  1.927 \\
& 85&0.5& (25, 8,  0.377)&  0.548& 14.298&  2.034&  5.394& 1.905\\
&&1&   (34, 10, 0.393)& 0.590& 18.158&  2.623&  6.673&  1.860\\
\cline{2-9}
&&0&(15,  6,  0.324)& 0.483&  8.639&  1.404 & 4.334&  1.927\\
& 95&0.5& (25, 8,  0.377)&  0.548& 14.298&  2.034&  5.394&  1.905\\
&&1&   (34, 10, 0.393)& 0.590& 18.158&  2.623&  6.673&  1.860\\
\hline

\end{tabular}
\label{NUM_COST_TAB_2}
\end{table}
 
\section{Numerical Illustration} \label{Exam}
\indent In this section, we first demonstrate the application of the proposed methodology in \S \ref{appl}. Subsequently, we conduct a simulation experiment to study the finite sample properties of the RASPs in \S \ref{Sim_Dem}.
\subsection{Illustrative Application}\label{appl}
In this section, we demonstrate the application of the RASPs using the background details of a real-life example discussed by \citet{mendenhall1958estimation}. The example pertains to the ARC-1 VHF communication transmitter-receivers and involves two failure modes (i.e., confirmed and unconfirmed failures). For convenience, we recalibrate the related data set by converting the failure times to thousands of hours. \\
\indent Now, we fit both independent and dependent competing risk models for this data set. As discussed in Section \ref{Frail_Mod}, these models assume that the shape parameters are equal. Toward this end, we first examine the validity of this assumption. We obtain the MLEs and the corresponding standard errors (SEs) of the model parameters for the independent and dependent competing risk models with equal and unequal shape parameters by maximizing the resulting log-likelihood function as in (\ref{log_like_eqn}). Table \ref{mod_sum} reports the MLEs and the corresponding standard errors (SEs) of the model parameters for both independent and dependent competing risk models with equal and unequal shape parameters. Furthermore, we also report the log-likelihood value, Akaike's Information Criterion (AIC), and Bayesian information criterion (BIC) in Table \ref{mod_sum}.\\
\indent Looking at the AIC and BIC values, it is obvious that both the dependent competing risks models outperform the independent competing risks models. However, it is not possible to pick a clear winner between the dependent competing risk models. As suggested in the literature in such contexts (see, for example, \citet{Yang_05}), the BIC is typically considered for model selection. Thus, it seems that the dependent model with equal shape parameters is sufficient to model this data set. For the purpose of illustration, we also consider the independent competing risks model with equal shape parameters, primarily for comparative purposes. \\  
\begin{table}[hbt!]
  \centering
  \caption{MLEs (SEs) and Goodness-of-fit for Different Models}
    \begin{tabular}{|c|cc|cc|}
    \hline
          & \multicolumn{4}{c|}{Model} \\
          \cline{2-5}
          & \multicolumn{2}{c|}{Independent } & \multicolumn{2}{c|}{Dependent} \\
          \cline{2-3}\cline{4-5}
          & \multicolumn{1}{c}{Equal} & \multicolumn{1}{c|}{Unequal} & \multicolumn{1}{c}{Equal} & \multicolumn{1}{c|}{Unequal} \\
          \hline
    $\eta_1$ &  0.439 (0.026)     &  0.436 (0.025)     &    0.303 (0.039)   &  0.299 (0.037) \\
    $\eta_2$ &   0.822 (0.076)    &   0.894 (0.107)    &    0.497 (0.086)   &  0.514 (0.096)\\
    $\gamma_1$ &  1.135 (0.052)   &    1.194 (0.066)   &     1.436 (0.130)  &  1.529 (0.145)\\
    $\gamma_2$ &  -     &  1.029 (0.083)     &    -   &  1.311 (0.142) \\
    $\nu$    &    -   &   -    &   0.616 (0.237)    &  0.646 (0.240) \\
    \hline
    log-likelihood &  $-138.404$     &  $-137.214$     &  $-134.124$     &  $-132.571$\\
    AIC   &   282.807    &   282.428    &   276.248    &  275.142\\
    BIC   &   294.539    &   298.071    &   291.891    &  294.696\\
    \hline
    \end{tabular}%
  \label{mod_sum}%
\end{table}\
\indent We now follow a similar approach to that discussed in the previous section to obtain $S_0$ and $S_1$, which are required to obtain the optimal sampling plans. In particular, we consider the specific time point of interest $t_0$ to be 0.15 and use the same discrimination ratio $d=1.5$ for both causes. Subsequently, we determine the optimal inspection interval by minimizing the $c$-optimality criterion, given by (\ref{des_crit_2}), for PIC-I schemes with 4, 6, and 8 inspection points. We further consider three cases\textemdash no intermediate withdrawal (i.e., $p=0$), 20\% withdrawal of surviving units (i.e., $p=0.2$), and 30\% withdrawal of surviving units (i.e., $p=0.3$).\\
\indent We report the optimal inspection interval ($h$), sample size ($n$), and acceptance limit ($\pi_c$) in Table \ref{Table_DA_2} for both independent and dependent competing risk models. It is interesting to note that the inspection interval ($h$) is relatively shorter for the independent model, suggesting that more frequent inspections are necessary for the independent model.  Furthermore, the acceptance limit is slightly on the higher side for the independent model compared to its dependent counterpart. In addition, the required sample size ($n$) increases for both models as the withdrawal proportion increases. \\ 
\begin{table}[hbt!]
    \centering
\small    \caption{Optimal Plans without Budgetary Constraints for Independent (I) and Dependent (D) Competing Risks Models}
    \begin{tabular}{|c|c|cccc|cccc|cccc|}
    \hline
 \multirow{3}{*}{$p_0$}& \multirow{3}{*}{Model} & \multicolumn{4}{|c|}{$M_0=4$} & \multicolumn{4}{c|}{$M_0=6$}& \multicolumn{4}{c|}{$M_0=8$}\\
    \cline{3-14}
& &\multirow{2}{*}{$h^*$}&$10\times$& \multirow{2}{*}{$n^*$}&\multirow{2}{*}{$\pi_c$}&\multirow{2}{*}{$h^*$}&$10\times$&\multirow{2}{*}{$n^*$}&\multirow{2}{*}{$\pi_c$}&\multirow{2}{*}{$h^*$}&$10\times $&\multirow{2}{*}{$n^*$}&\multirow{2}{*}{$\pi_c$}\\
&&&$\phi(\boldsymbol{\boldsymbol{\zeta}}^*)$&&&&$\phi(\boldsymbol{\boldsymbol{\zeta}}^*)$&&&&$\phi(\boldsymbol{\boldsymbol{\zeta}}^*)$&&\\
     \hline
 \multirow{2}{*}{0}   &  I & 0.064 & 1.750 &72& 0.563 &0.048 & 1.690 &69&  0.563&  0.039&  1.659& 68&  0.563\\
 &    D&0.107 & 2.016& 71&  0.538&0.089&  1.918 &68&  0.538& 0.078 & 1.875& 66&  0.538\\
 
    \hline
 \multirow{2}{*}{0.2}      &      I& 0.083&  1.914& 78&  0.563&  0.078& 1.914 &78&  0.563& 0.077&  1.906& 78&  0.563\\
  &   D& 0.121 & 2.116& 74&  0.538&0.112&  2.070& 73&  0.538& 0.109&  2.052& 72&  0.538\\
     \hline
  \multirow{2}{*}{0.3}     &I& 0.091& 1.982& 81& 0.563&0.090&1.979& 81& 0.563 &0.090&   1.979&81&0.563\\
     &D&0.126&2.158&75&0.538&0.121&2.124&74&0.538&0.119&2.115&74&0.538\\
     \hline
     \end{tabular}
    \label{Table_DA_2}
\end{table}
\indent We further consider the budgetary constraints and obtain the optimal RASPs as discussed in \S \ref{OPT_WC}. For the purpose of illustration, we assume the cost components $C_S=0.1$, $c_\tau=5$, $C_D=0.025$, and $C_I=10$. In addition, we consider the available budget as $C_B=65$. Table \ref{optimal1} reports the optimal sample size ($n$), number of inspections ($M$), inspection interval ($h$) along with the acceptance limit ($\pi_c$). Furthermore, Table \ref{optimal1} also provides the expected number of failures ($E[D^*]$), the expected duration of the experiment ($E[\tau^*]$), the expected number of inspections ($E[I^*]$) and the optimal value of the design criterion ($\phi(\zeta^*)$). As observed in the previous case, the acceptance limit is slightly on the higher side for the independent model compared to its dependent counterpart. Also, the inspection interval is slightly longer for the dependent model. Although there is hardly any impact on the number of inspections, the required sample size is slightly on the lower side for the dependent model as the withdrawal proportion increases. Interestingly, the expected number of failures is on the higher side for the dependent model. Moreover, the expected duration seems to be on the higher side for the dependent model.\\
\begin{table}[hbt!]
    \centering
      \caption{Optimal Plans with Budgetary Constraints for Independent (I) and Dependent (D) Competing Risks Models}
    \begin{tabular}{|cccccccc|}
    \hline
$p_0$& Model & $(n^*,M^*,h^*)$&$\pi_c$&$E[D^*]$&$E[\tau^*]$&$E[I^*]$& $10 \times \phi(\boldsymbol{\boldsymbol{\zeta}}^*)$\\
    \hline
 \multirow{2}{*}{0.0}& I& (70,  5, 0.054)& 0.563&  40.337&  0.270&  5.000&  1.713\\
& D&(69,  5,  0.096)&  0.538& 55.867&  0.480&  5.000&  1.961\\
\hline
\multirow{2}{*}{0.2}& I &(78,  5,  0.079)&  0.563& 41.432& 0.395&  5.000&  1.907\\
&D &(73,  5,  0.115) &  0.538& 48.605&  0.575&  4.998&  2.084\\
\hline
\multirow{2}{*}{0.3}&I& (81, 5, 0.091)& 0.563& 40.640&  0.455&  4.998&  1.979\\
&D&(75,  5,  0.122)&  0.538&  45.649&  0.606&  4.967&  2.136\\
\hline
    \end{tabular}
    \label{optimal1}
\end{table}
\indent Toward this end, we briefly demonstrate the process of making a decision regarding the acceptance or rejection of a lot with a simulated data set. We consider the optimal RASP $(73, 5, 0.115)$ along with $p=0.2$ for the purpose of illustration. In this plan, the manufacturer performs a life test with a sample size of $73$. A total of 5 inspections are carried out at an interval of $0.115$ time units. Furthermore, 20\% of the surviving units are withdrawn at each intermediate inspection, with all remaining serving items being removed from the experiment at the 5th inspection. Based on the observed data set, a decision is taken regarding the acceptance or rejection of the lot using the estimated reliability of the unit. Thus, in this case, a lot is accepted if $\widehat{\overline{F}}_T(0.15\ | \ \widehat{\boldsymbol{\theta}})> 0.538$. Using the MLEs presented above, we simulate a data set following the algorithm provided by \citet{ROY_BP_23}. The simulated data set is presented in Table \ref{generated data}.
For this data set, we obtain the MLEs of the model parameters and r.f., which are as follows:
$$\widehat{\eta}_1=0.292, \widehat{\eta}_2=0.374, \widehat{\gamma}=1.779,
\widehat{\nu}=0.668,\,\,\text{and}\,\, \widehat{\overline{F}}_T(0.15\ | \ \widehat{\boldsymbol{\theta}})=0.648.$$
Since the estimated reliability is above the acceptance limit of 0.538, we accept the lot.
 \begin{table}[hbt!]
     \centering
       \caption{Simulated PIC-I Data Set}
     \begin{tabular}{ccccc}
     \hline
        $i$  & Time interval & $d_{i1}$ &$d_{i2}$&$r_i$ \\
        \hline
        1  & (0, 0.115] & 11&7&11\\
        2& (0.115, 0.230]& 10 & 8 & 5\\
3& (0.230, 0.345] &  6 & 4 & 2 \\
4& (0.345, 0.460] & 1 & 0 & 1 \\ 
5 &(0.460, 0.575] & 3 & 1 & 3 \\
\hline
     \end{tabular}
     \label{generated data}
 \end{table}   
\subsection{Simulation Evaluation} \label{Sim_Dem}
\indent In this sub-section, we evaluate the finite sample behavior of the optimal RASPs through a simulation study. For brevity, we consider the optimal RASPs under budgetary constraints for subsequent analysis.\\
\indent For each optimal PIC-I scheme, we simulate 5000 data sets following the algorithm provided by \citet{ROY_BP_23}. In each case, we consider the MLEs reported in the previous subsection as the true values (TVs) of the model parameters. For each data set, we obtain the MLEs of the relevant parameters by maximizing the log-likelihood function in (\ref{log_like_eqn}). Using these MLEs, we then obtain the MLE of the r.f. $\overline{F}_T(t|\boldsymbol{\theta})$, given by (\ref{surv}) and ({\ref{surv1}}). The histograms of the resulting MLEs of $\overline{F}_T(t|\boldsymbol{\theta})$ for both independent and dependent competing risks models with withdrawal proportions $p=0$, 0.2 and 0.3 are given in Figure \ref{fig:prob1_6_1}, Figure \ref{fig:prob1_6_2} and Figure \ref{fig:prob1_6_3}, respectively. As expected, the histograms are symmetric in nature, which is in line with our observation in Theorem \ref{th2}. \\
\begin{figure}[hbt!]
\centering
\begin{subfigure}{0.45\textwidth}
\centering
\includegraphics[width=\textwidth, height=0.25\textheight]{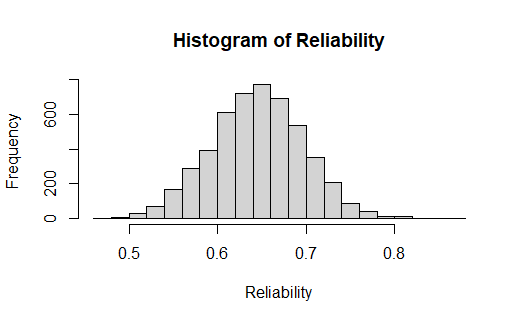}
\caption{Independent}
\label{fig:prob1_6_2.0}
\end{subfigure}
\begin{subfigure}{0.45\textwidth}
\centering
\includegraphics[width=\textwidth, height=0.25\textheight]{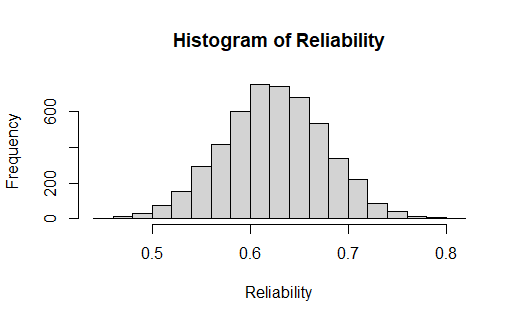}
\caption{Dependent}
\label{fig:prob1_6_1.0}
\end{subfigure}
\caption{Histograms of $\widehat{\overline{F}}_T(t|\boldsymbol{\theta})$ for Independent and Dependent Models when $p=0$}
\label{fig:prob1_6_1}
\end{figure}
\begin{figure}[hbt!]
\centering
\begin{subfigure}{0.45\textwidth}
\centering
\includegraphics[width=\textwidth, height=0.25\textheight]{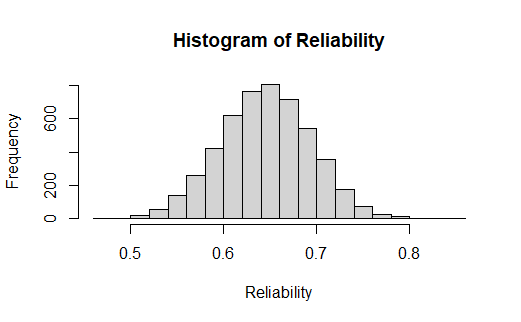}
\caption{Independent}
\label{fig:prob1_6_2.2}
\end{subfigure}
\begin{subfigure}{0.45\textwidth}
\centering
\includegraphics[width=\textwidth, height=0.25\textheight]{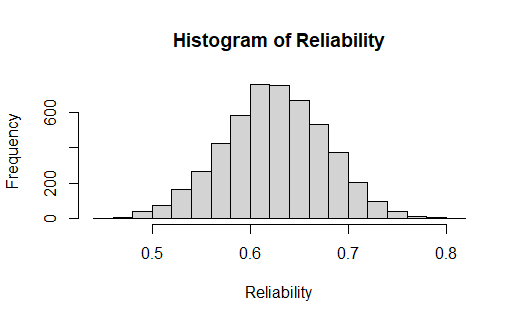}
\caption{Dependent}
\label{fig:prob1_6_1.2}
\end{subfigure}
\caption{Histograms of $\widehat{\overline{F}}_T(t|\boldsymbol{\theta})$ for Independent and Dependent Models when $p=0.2$}
\label{fig:prob1_6_2}
\end{figure}

\begin{figure}[hbt!]
\centering
\begin{subfigure}{0.45\textwidth}
\centering
\includegraphics[width=\textwidth, height=0.25\textheight]{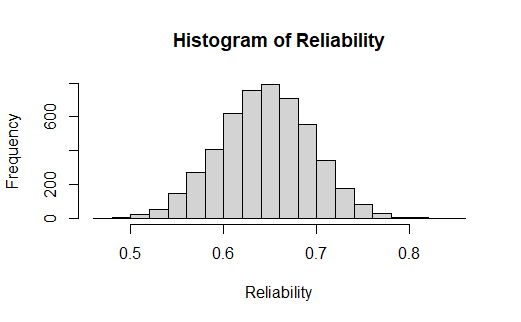}
\caption{Independent}
\label{fig:prob1_6_2.3}
\end{subfigure}
\begin{subfigure}{0.45\textwidth}
\centering
\includegraphics[width=\textwidth, height=0.25\textheight]{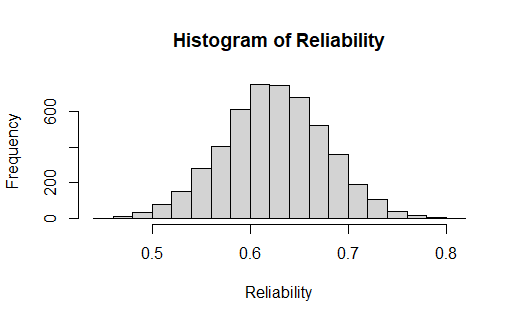}
\caption{Dependent}
\label{fig:prob1_6_1.3}
\end{subfigure}
\caption{Histograms of $\widehat{\overline{F}}_T(t|\boldsymbol{\theta})$ for Independent and Dependent Models when $p=0.3$}
\label{fig:prob1_6_3}
\end{figure}
\indent We further obtain the average estimate of the r.f. and the corresponding root mean square deviation (RMSD) over 5000 simulation runs. Table \ref{simulation} reports the average estimate of the r.f. and the corresponding RMSD along with their respective True Values (TVs). We also compute the estimate of the standardized variance $S^2$ in each case using the MLEs of the model parameters. Table \ref{simulation} also provides the average estimate of standardized variance $S^2$ and the corresponding RMSD. Furthermore, we check whether the lot is accepted or rejected in each case using the decision rule presented in Section \ref{dev} and obtain the estimates of the producer's risk $(\alpha)$ and the consumer's risk $(\beta)$, as given by (\ref{SP1}) and (\ref{SP2}), respectively. Table \ref{simulation} reports the estimated producer's risk ($\widehat{\alpha}$) and the consumer's risk ($\widehat{\beta}$). It is easy to note that the average estimates of the r.f. are close to their TVs in each case. Furthermore, the precision in reliability estimates is hardly affected as the withdrawal proportion increases. A similar observation can also be drawn for $S^2$. We also note that the estimates of the producer's and consumer's risk are reasonably close to their respective specified values.

\begin{table}[hbt!]
    \small  \centering
      \caption{Average Estimates and RMSDs of $\overline{F}_T(t|\boldsymbol{\theta})$ and $S^2$ along with Estimates of Producer's and Consumer's Risks for Independent (I) and Dependent (D) Competing Risks Models}
    \begin{tabular}{|c|c|cccccccc|}
    \hline
      &TV && & Avg. & $100 \times$ RMSD  & $10 \times$ Avg. & $100\times$ RMSD  & &\\
      Model  &of R.F.&$p_0$&$\boldsymbol{\zeta}^*$& R.F. & of R.F. & $S^2$&of $S^2$&$\widehat{\alpha}$&$\widehat{\beta}$\\
     \hline
      \multirow{3}{*}{I}   &  \multirow{3}{*}{0.644}  &0.0&(70, 5, 0.054)&0.645&4.886&1.696&1.122&$0.049$&0.128\\
         \cline{3-10}
      &  &0.2&(78, 5, 0.079)&0.645&4.926&1.889&1.156&$0.051$&0.127\\
      \cline{3-10}
      &&0.3& (81, 5, 0.091) &0.645&4.933&1.959&1.199&0.050&0.123\\
        \hline
         \multirow{3}{*}{D}&\multirow{3}{*}{0.626} &0.0&(69, 5, 0.096)&0.624&5.200&1.957&1.564&$0.050$&0.111\\
         \cline{3-10}
          &&0.2&(73, 5, 0.115)&0.625&5.250&2.081&1.555&$0.054$&0.120\\
           \cline{3-10}
          &&0.3&(75, 5, 0.122)&0.625&5.250&2.131&1.541&0.050&0.121\\
         \hline
    \end{tabular}
    \label{simulation}
 \end{table}

\section{Conclusion}\label{con}
\indent In this work, we consider the development of RASPs under PIC-I schemes in the presence of competing causes of failure. At the outset, we present a general framework for incorporating the presence of competing risks under PIC-I schemes and obtain the Fisher information matrix in this context. We then discuss the asymptotic properties of the MLEs of the model parameters under a set of regularity conditions. Note that such asymptotic properties are essential for developing RASPs. The general framework presented in this work can accommodate both dependent and independent competing risks models and can prove to be really useful in modeling the interval-censored data sets, which are common in several managerial applications.\\
\indent Subsequently, we present a frailty model as a special case of the general framework developed in this work. The frailty model incorporates the dependence among the potential failure times, which is largely neglected in designing the optimal censoring schemes. Moreover, this model allows us to obtain a popular independent competing risks model as a special case, thus enabling us to study the effect of incorporating the dependence structure among the potential failure times on the RASPs. Subsequently, we develop the RASPs in three important contexts for both the dependent and independent competing risks models. First, we present the traditional RASPs using the producer's and consumer's risks. These traditional RASPs require several specific information, including details such as the inspection times and the proportion of intermediate withdrawals, which the decision-makers must decide suitably while conducting the life-test experiments under PIC-I scenes. We thus consider a suitable $c$-optimal design criterion which allows us to determine these design variables. We develop the RASPs both with and without budgetary constraints. Through a detailed numerical study, we find out the level of dependence between the potential failure times has a significant impact on the resulting RASPs. We further demonstrate the application of the developed RASPs in a practical context and study the finite-sample properties of these RASPs through a detailed simulation study. We observe that these RASPs exhibit the desired asymptotic properties even for a modest sample size.\\
\indent This work can be extended in a number of important ways. First, we use a frailty model to incorporate the dependence among the potential failure times. However, the same can also be modeled using popular copula models \citep{LO2025108104}. It will be interesting to develop the RASPs for such models. Second, as discussed by \citet{roy_19}, important prior information is often available in practice, and in such cases, a Bayesian framework may be considered to develop Bayesian RASPs.

\bibliographystyle{apalike}
\bibliography{BP.bib} 
\appendix
\section{Proof of Lemma \ref{l1}:}\label{l1app}
\begin{proof}
By differentiating (\ref{log_like_eqn}) with respect to $\theta_u$, we have
\begin{align}\label{ul}
     \frac{\partial l(\boldsymbol{\theta}\ | \ \mathcal{D}_M,\boldsymbol{\boldsymbol{\zeta}})}{\partial\theta_u}=\sum_{i=1}^M\sum_{k=1}^{n_i}\left[\sum_{j=1}^J\left(\frac{\delta_{ikj}} {q_{ij}}\right)\left(\frac{\partial q_{ij}}{\partial\theta_u}\right)-\frac{(1-{\delta_{ik+}})}{(1-q_{i})}\left(\frac{\partial q_{i}}{\partial\theta_u}\right)\right]=\sum_{i=1}^M\sum_{k=1}^{n_i}U_{ik,u},
\end{align}
where $U_{ik,u}=\sum_{j=1}^J\left(\frac{\delta_{ikj}} {q_{ij}}\right)\left(\frac{\partial q_{ij}}{\partial\theta_u}\right)-\frac{(1-{\delta_{ik+}})}{(1-q_{i})}\left(\frac{\partial q_{i}}{\partial\theta_u}\right)$ and $\delta_{ik+}=\sum_{j=1}^J  \delta_{ikj}$. Then, we get 
\begin{align*}
    E\left[ \frac{\partial l(\boldsymbol{\theta}\ | \ \mathcal{D}_M,\boldsymbol{\boldsymbol{\zeta}})}{\partial\theta_u}\right]=\sum_{i=1}^M E_{N_i}\left[\sum_{k=1}^{N_i}E\left[U_{ik,u}\ | \ N_i\right]\right].
\end{align*}
It is sufficient to prove that $E\left[U_{ik,u}\ | \ N_i\right]=0$.
We know $E[\delta_{ikj}]=q_{ij}$ and $E[1-\delta_{ik+}]=(1-q_i)$. Thus, we have 
\begin{align*}
    E\left[U_{ik,u}\ | \ N_i\right]=\sum_{j=1}^J\frac{\partial q_{ij}}{\partial\theta_u}-\frac{\partial q_{i}}{\partial\theta_u}.
\end{align*}
We know $\sum_{j=1}^Jq_{ij}=q_i$. This implies that
\begin{align*}
    \sum_{j=1}^J\frac{\partial q_{ij}}{\partial\theta_u}=\frac{\partial q_{i}}{\partial\theta_u}.
\end{align*}
Hence, $E\left[U_{ik,u}\ | \ N_i\right]=0$.  
\end{proof}
\section{Proof of Lemma \ref{lemma_2}:}\label{l2app}
\begin{proof}
By differentiating (\ref{ul}) with respect to $\theta_v$, we have
\begin{align}\label{lemma_2_dl}
     \frac{\partial^2 l(\boldsymbol{\theta}\ | \ \mathcal{D}_M,\boldsymbol{\boldsymbol{\zeta}})}{\partial\theta_u\partial\theta_v}=&\sum_{i=1}^M \sum_{k=1}^{n_i}\left[\sum_{j=1}^J\left(\frac{\delta_{ikj}} {q_{ij}}\right)\left(\frac{\partial^2 q_{ij}}{\partial\theta_u\partial\theta_v}\right)-\frac{(1-{\delta_{ik+}})}{(1-q_{i})}\left(\frac{\partial^2 q_{i}}{\partial\theta_u\partial\theta_v}\right)\right]\nonumber\\
     &-\sum_{i=1}^M \sum_{k=1}^{n_i}\left[\sum_{j=1}^J\left(\frac{\delta_{ikj}} {q_{ij}^2}\right)\left(\frac{\partial q_{ij}}{\partial\theta_u}\right)\left(\frac{\partial q_{ij}}{\partial\theta_v}\right)+\frac{(1-{\delta_{ik+}})}{(1-q_{i})^2}\left(\frac{\partial q_{i}}{\partial\theta_u}\right)\left(\frac{\partial q_{i}}{\partial\theta_v}\right)\right]\\
    = &\sum_{i=1}^M \sum_{k=1}^{n_i} V_{ik,uv}\nonumber, ~\text{say},
\end{align}
where $V_{ik,uv}=\left[\sum\limits_{j=1}^J\left(\frac{\delta_{ikj}} {q_{ij}}\right)\left(\frac{\partial^2 q_{ij}}{\partial\theta_u\partial\theta_v}\right)-\frac{(1-{\delta_{ik+}})}{(1-q_{i})}\left(\frac{\partial^2 q_{i}}{\partial\theta_u\partial\theta_v}\right)\right]-\left[\sum\limits_{j=1}^J\left(\frac{\delta_{ikj}} {q_{ij}^2}\right)\left(\frac{\partial q_{ij}}{\partial\theta_u}\right)\left(\frac{\partial q_{ij}}{\partial\theta_v}\right)+\frac{(1-{\delta_{ik+}})}{(1-q_{i})^2}\left(\frac{\partial q_{i}}{\partial\theta_u}\right)\left(\frac{\partial q_{i}}{\partial\theta_v}\right)\right].$
Now, proceeding in the same manner as in Lemma 1, we have
\begin{align*}
 \sum_{i=1}^M E_{N_i}\left[\sum_{k=1}^{N_i}   E\left[\sum_{j=1}^J\left(\frac{\delta_{ikj}} {q_{ij}}\right)\left(\frac{\partial^2 q_{ij}}{\partial\theta_u\partial\theta_v}\right)-\frac{(1-{\delta_{ik+}})}{(1-q_{i})}\left(\frac{\partial^2 q_{i}}{\partial\theta_u\partial\theta_v}\right)\right]\right]=0.
\end{align*}
Therefore, taking the expectation of (\ref{lemma_2_dl}), we have
\begin{align*}
    E\left[ \frac{\partial^2 l(\boldsymbol{\theta}\ | \ \mathcal{D}_M,\boldsymbol{\boldsymbol{\zeta}})}{\partial\theta_u\partial\theta_v}\right]=-\sum_{i=1}^M E_{N_{i}}\Bigg[\sum_{k=1}^{N_i}\Bigg[\sum_{j=1}^J\left(\frac{E[\delta_{ikj}]} {q_{ij}^2}\right)\left(\frac{\partial q_{ij}}{\partial\theta_u}\right)\left(\frac{\partial q_{ij}}{\partial\theta_v}\right)\nonumber\\+\left(\frac{E[1-\delta_{ik+}]}{(1-q_{i})^2}\right)\left(\frac{\partial q_{i}}{\partial\theta_u}\right)\left(\frac{\partial q_{i}}{\partial\theta_v}\right)\Bigg]\Bigg].
\end{align*}
Now, using $E[\delta_{ikj}]=q_{ij}$ and $E[1-\delta_{ik+}]=1-q_i$, we get
\begin{align}\label{fish}
    E\left[- \frac{\partial^2 l(\boldsymbol{\theta}\ | \ \mathcal{D}_M,\boldsymbol{\boldsymbol{\zeta}})}{\partial\theta_u\partial\theta_v}\right]=\sum_{i=1}^M \left[\sum_{j=1}^J\left(\frac{E[N_i]} {q_{ij}}\right)\left(\frac{\partial q_{ij}}{\partial\theta_u}\right)\left(\frac{\partial q_{ij}}{\partial\theta_v}\right)+\left(\frac{E[N_i]}{(1-q_{i})}\right)\left(\frac{\partial q_{i}}{\partial\theta_u}\right)\left(\frac{\partial q_{i}}{\partial\theta_v}\right)\right].
\end{align}
\indent Note that
\begin{align*}
    \left(\frac{\partial l(\boldsymbol{\theta}\ | \ \mathcal{D}_M,\boldsymbol{\boldsymbol{\zeta}})}{\partial\theta_u}\right)\left(\frac{\partial l(\boldsymbol{\theta}\ | \ \mathcal{D}_M,\boldsymbol{\boldsymbol{\zeta}})}{\partial\theta_v}\right)
=\sum_{i=1}^M\sum_{k=1}^{n_i}\sum_{i'=1}^M\sum_{k'=1}^{n_{i'}}U_{ik,u}U_{i'k',v},
\end{align*}
where $U_{ik,u}$ and $U_{i'k',v}$ are as defined above. Now, suppose $i<i'$. Then, for given $n_{i'}$, $n_i$ is fixed and hence $E[U_{ik,u}U_{i'k',v}\ | \ n_{i'}]=U_{ik,u}E[U_{i'k',v}\ | \ n_{i'}]=0.$ Similarly, when $i>i'$, $E[U_{ik,u}U_{i'k',v}\ | \ n_{i}]=U_{i'k',v}E[U_{ik,u}\ | \ n_{i}]=0.$
Also, when $i'=i$ and $k'\neq k$, we have
\begin{align*}
   & E[U_{ik,u}U_{ik',v}\ | \ n_{i}]\\=&E\left[\left\{\sum_{j=1}^J\left(\frac{\delta_{ikj}} {q_{ij}}\right)\left(\frac{\partial q_{ij}}{\partial\theta_u}\right)-\frac{(1-{\delta_{ik+}})}{(1-q_{i})}\left(\frac{\partial q_{i}}{\partial\theta_u}\right)\right\}\left\{\sum_{j'=1}^J\left(\frac{\delta_{ik'j'}} {q_{ij'}}\right)\left(\frac{\partial q_{ij}}{\partial\theta_v}\right)-\frac{(1-{\delta_{ik'+}})}{(1-q_{i})}\left(\frac{\partial q_{i}}{\partial\theta_v}\right)\right\}\right]\\
   =&E\left[\sum_{j=1}^J\sum_{j'=1}^J\left(\frac{\delta_{ikj}} {q_{ij}}\right)\left(\frac{\delta_{ik'j'}} {q_{ij'}}\right)\left(\frac{\partial q_{ij}}{\partial\theta_u}\right)\left(\frac{\partial q_{ij}}{\partial\theta_v}\right)-\sum_{j'=1}^J\left(\frac{\delta_{ik'j'}} {q_{ij'}}\right)\left(\frac{\partial q_{ij}}{\partial\theta_v}\right)\frac{(1-{\delta_{ik+}})}{(1-q_{i})}\left(\frac{\partial q_{i}}{\partial\theta_u}\right)\right.\\
   &\left.~~~~~~-\sum_{j=1}^J\left(\frac{\delta_{ikj}} {q_{ij}}\right)\left(\frac{\partial q_{ij}}{\partial\theta_u}\right)\frac{(1-{\delta_{ik'+}})}{(1-q_{i})}\left(\frac{\partial q_{i}}{\partial\theta_v}\right)+\frac{(1-{\delta_{ik+}})}{(1-q_{i})}\frac{(1-{\delta_{ik'+}})}{(1-q_{i})}\left(\frac{\partial q_{i}}{\partial\theta_u}\right)\left(\frac{\partial q_{i}}{\partial\theta_v}\right)\right]=0.
\end{align*}
Since $T_1,\ldots,T_n$ are i.i.d., we have, for any fixed $i$, 
$$E\left[\sum_{j=1}^J\delta_{ikj}\sum_{j'=1}^J\delta_{ik'j'}\right]=E\left[\sum_{j=1}^J\delta_{ikj}\right]E\left[\sum_{j'=1}^J\delta_{ik'j'}\right],$$ $$E\left[\sum_{j=1}^J\delta_{ikj}(1-\delta_{ik'+})\right]=E\left[\sum_{j=1}^J\delta_{ikj}\right]E\left[1-\delta_{ik'+})\right],$$
$$E\left[\sum_{j'=1}^J\delta_{ik'j'}(1-\delta_{ik+})\right]=E\left[\sum_{j'=1}^J\delta_{ik'j'}\right]E\left[1-\delta_{ik+})\right],$$
and
$$E\left[(1-\delta_{ik+})(1-\delta_{ik'+})\right]=E\left[(1-\delta_{ik+})\right]E\left[1-\delta_{ik'+})\right].$$
Now, it is easy to see that
$ E[U_{ik,u}U_{ik',v}\ | \ n_{i}]=0.$
Lastly, when $i'=i$ and $k'=k$, we have
\begin{align}\label{lemma_2_2}
   & E[U_{ik,u}U_{ik,v}\ | \ n_{i}] \nonumber\\
   =&E\left[\sum_{j=1}^J\sum_{j'=1}^J\left(\frac{\delta_{ikj}} {q_{ij}}\right)\left(\frac{\delta_{ikj'}} {q_{ij'}}\right)\left(\frac{\partial q_{ij}}{\partial\theta_u}\right)\left(\frac{\partial q_{ij}}{\partial\theta_v}\right)-\sum_{j'=1}^J\left(\frac{\delta_{ikj'}} {q_{ij'}}\right)\left(\frac{\partial q_{ij'}}{\partial\theta_v}\right)\frac{(1-{\delta_{ik+}})}{(1-q_{i})}\left(\frac{\partial q_{i}}{\partial\theta_u}\right)\right.\nonumber\\
   &\left.~~~~~~-\sum_{j=1}^J\left(\frac{\delta_{ikj}} {q_{ij}}\right)\left(\frac{\partial q_{ij}}{\partial\theta_u}\right)\frac{(1-{\delta_{ik+}})}{(1-q_{i})}\left(\frac{\partial q_{i}}{\partial\theta_v}\right)+\frac{(1-{\delta_{ik+}})^2}{(1-q_{i})^2}\left(\frac{\partial q_{i}}{\partial\theta_u}\right)\left(\frac{\partial q_{i}}{\partial\theta_v}\right)\right].
\end{align}
Note that $P(\delta_{ikj}\delta_{ikj'}=1)=0$ for $j\neq j'$, and $P(\delta_{ikl}(1-\delta_{ik+})=1)=0$ for $l=1,\ldots,J$.
Therefore, for $i=1,\ldots,M$, $k=1,\ldots,n_i$ and $j,j'=1,\ldots,J$, we get
\begin{align*}
    E[\delta_{ikj}\delta_{ikj'}\ | \ n_i]=\begin{cases}
        E[\delta_{ikj}^2\ | \ n_i]= E[\delta_{ikj}\ | \ n_i]=q_{ij} &j=j'\\
        0&j\neq j'
    \end{cases},
\end{align*}
\begin{align*}
    E[\delta_{ikl}(1-\delta_{ik+})\ | \ n_i]=0, &~~~~~\text{for }l=j,j',
\end{align*}
and
\begin{align*}
       E[(1-\delta_{ik+})^2\ | \ n_i]= E[(1-\delta_{ikj})\ | \ n_i]=1-q_{i}.
\end{align*}
Then, (\ref{lemma_2_2}) can be written as
\begin{align}\label{var}
    E[U_{ik,u}U_{ik,v}\ | \ n_{i}]=\sum_{j=1}^J\frac{1}{q_{ij}}\left(\frac{\partial q_{ij}}{\partial\theta_u}\right)\left(\frac{\partial q_{ij}}{\partial\theta_v}\right)+\frac{1}{(1-q_i)}\left(\frac{\partial q_{i}}{\partial\theta_u}\right)\left(\frac{\partial q_{i}}{\partial\theta_v}\right)=\sigma^2 _{i,uv}\,\,\text{(say).}
\end{align}

Thus, we have
\begin{align}\label{Inform_uv}
    &E\left[\left(\frac{\partial l(\boldsymbol{\theta}\ | \ \mathcal{D}_M,\boldsymbol{\boldsymbol{\zeta}})}{\partial \theta_u}\right)\left(\frac{\partial l(\boldsymbol{\theta}\ | \ \mathcal{D}_M,\boldsymbol{\boldsymbol{\zeta}})}{\partial \theta_v}\right)\right]\nonumber\\
    &=\sum_{i=1}^M E_{N_i}\left[U_{ik,u}U_{ik,v}\ | \ N_i\right]\nonumber\\
    &=\sum_{i=1}^M \left[\sum_{j=1}^J\frac{E[N_i]}{q_{ij}}\left(\frac{\partial q_{ij}}{\partial\theta_u}\right)\left(\frac{\partial q_{ij}}{\partial\theta_v}\right)+\frac{E[N_i]}{(1-q_i)}\left(\frac{\partial q_{i}}{\partial\theta_u}\right)\left(\frac{\partial q_{i}}{\partial\theta_v}\right)\right].
\end{align}
\end{proof}\\
\section{Proof of Theorem \ref{th2}}\label{th2app}

\indent We first present the following results, which are necessary to prove Theorem 2. For convenience, we use $F_T(t)$, $\overline{F}_T(t)$, $G(j,t)$, and $l(\boldsymbol{\theta})$ instead of $F_T(t\ | \ \boldsymbol{\theta})$, $\overline{F}_T(t|\boldsymbol{\theta})$, $G(j,t | \boldsymbol{\theta})$ and  $l(\boldsymbol{\theta}\ | \ \mathcal{D}_m,\boldsymbol{\boldsymbol{\zeta}})$, respectively.

\begin{lemma}\label{l3}
Under the regularity conditions (I)-(VII), the first-, second-, and third-order derivatives of $q_{ij}$ and $q_{i}$ are bounded.
\end{lemma}
\begin{proof}
\indent We have
  \begin{align}
        \frac{\partial q_{ij}}{\partial\theta_u}=\frac{1}{\overline{F}_T(L_{i-1})}\left[\frac{\partial\left[{G}(j,L_{i})-{G}(j,L_{i-1})\right]}{\partial\theta_u}+q_{ij}\left(\frac{\partial{F}_T(L_{i-1})}{\partial\theta_u}\right)\right], 
    \end{align}

    for $u=1,2\ldots,s$. Now, let us define the following bounds:
    \begin{align}\label{F}
       \kappa=\max_{1\leq i\leq M}\sup_{\boldsymbol{\theta}\in\boldsymbol{\Theta}_0}\frac{1}{\overline{F}_T(L_{i})}, 
    \end{align}
    \begin{align}\label{df}
        A_u=\max_{1\leq i\leq M}\sup_{\boldsymbol{\theta}\in\boldsymbol{\Theta}_0}\left|\frac{\partial F_T{(L_i)}}{\partial\theta_u}\right|
    \end{align}
    and
    \begin{align}\label{dg}
        B_{u}=\max_{1\leq i\leq M}\max_{1\leq j\leq J}\sup_{\boldsymbol{\theta}\in\boldsymbol{\Theta}_0}\left|\frac{\partial G{(j,L_i)}}{\partial\theta_u}\right|,
    \end{align}
    for $u=1,2,\ldots,s$. Then, by using (\ref{F}), (\ref{df}) and (\ref{dg}), we have
    \begin{align*}
        \left|\frac{\partial q_{ij}}{\partial\theta_u}\right|\leq\kappa[2B_{u}+A_u]=C_{u}, ~~\text{say. }
    \end{align*}
Thus, the first-order derivative of $q_{ij}$ is bounded. Now, differentiating with respect to $\theta_v$, for $v=1,2,\ldots,s$, we get
   \begin{align}\label{d2qij}
\frac{\partial^2 q_{ij}}{\partial\theta_u\partial\theta_v}&=\frac{1}{[\overline{F}_T(L_{i-1})]^2}\left[\frac{\partial \left[G(j,L_i)-G(j,L_{i-1})\right]}{\partial\theta_u}+q_{ij}\left(\frac{\partial F_T(L_i)}{\partial\theta_u}\right)\right]\left(\frac{\partial F_T(L_{i-1})}{\partial\theta_v}\right)\nonumber\\
&+\frac{1}{\overline{F}_T(L_{i-1})}\left[\frac{\partial^2 \left[G(j,L_i)-G(j,L_{i-1})\right]}{\partial\theta_u\partial\theta_v}+\left(\frac{\partial q_{ij}}{\partial\theta_v}\right)\left(\frac{\partial F_T(L_{i-1})}{\partial\theta_u}\right)+q_{ij}\left(\frac{\partial^2 F_T(L_{i-1})}{\partial\theta_u\partial\theta_v}\right)\right]
    \end{align}
    Under the regularity conditions, the second-order derivatives of $G(j, L_i)$ and $F_T(L_i)$ are assumed to be bounded. Therefore, we define
  \begin{align}\label{d2f}
        A_{uv}=\max_{1\leq i\leq M}\sup_{\boldsymbol{\theta}\in\boldsymbol{\Theta}_0}\left|\frac{\partial^2 F_T{(L_i)}}{\partial\theta_u\partial\theta_v}\right|
    \end{align}
    and
    \begin{align}\label{d2g}
        B_{uv}=\max_{1\leq i\leq M}\max_{1\leq j\leq J}\sup_{\boldsymbol{\theta}\in\boldsymbol{\Theta}_0}\left|\frac{\partial^2 G{(j,L_i)}}{\partial\theta_u\partial\theta_v}\right|,
    \end{align}
        for $u=1,2,\ldots,s$ and $v=1,2,\ldots,s$. Then, by using (\ref{F}), (\ref{df}),(\ref{dg}), (\ref{d2f}) and (\ref{d2g}) we have
          \begin{align*}
        \left|\frac{\partial^2 q_{ij}}{\partial\theta_u\partial\theta_v}\right|\leq\kappa^2[2B_{u}+A_u]A_v+\kappa[2B_{uv}+C_{u}A_u+A_{uv}]=C_{uv}, ~~\text{say.}
    \end{align*}
    Differentiating Equation (\ref{d2qij}) with respect to $\theta_w$, for $w=1,2,\ldots,s$, we get
 \small    \allowdisplaybreaks   \begin{align*}       
&\frac{\partial ^3q_{ij}}{\partial\theta_u\partial\theta_v\partial\theta_w}\\=&\frac{2}{[\overline{F}_T(L_{i-1})]^3}\left[\frac{\partial \left[G(j,L_i)-G(j,L_{i-1})\right]}{\partial\theta_u}+q_{ij}\left(\frac{\partial F_T(L_i)}{\partial\theta_u}\right)\right]\left(\frac{\partial F_T (L_{i-1})}{\partial\theta_v}\right)\left(\frac{\partial F_T(L_{i-1})}{\partial\theta_w}\right)\\
&+\frac{1}{[\overline{F}_T(L_{i-1})]^2}\left[\frac{\partial \left[G(j,L_i)-G(j,L_{i-1})\right]}{\partial\theta_u}+q_{ij}\left(\frac{\partial F_T(L_i)}{\partial\theta_u}\right)\right]\left(\frac{\partial^2 F_T(L_{i-1})}{\partial\theta_v\partial\theta_w}\right)\\
&+\frac{1}{[\overline{F}_T(L_{i-1})]^2}\left[\frac{\partial^2 \left[G(j,L_i)-G(j,L_{i-1})\right]}{\partial\theta_u\partial\theta_w}+q_{ij}\left(\frac{\partial^2  F_T(L_i)}{\partial\theta_u\partial\theta_w}\right)+\left(\frac{\partial q_{ij}}{\partial\theta_w}\right)\left(\frac{\partial  F_T(L_i)}{\partial\theta_u}\right)\right]\left(\frac{\partial F_T(L_{i-1})}{\partial\theta_v}\right)\\
&+\frac{1}{[\overline{F}_T(L_{i-1})]^2}\left[\frac{\partial^2 \left[G(j,L_i)-G(j,L_{i-1})\right]}{\partial\theta_u\partial\theta_v}+\left(\frac{\partial q_{ij}}{\partial\theta_v}\right)\left(\frac{\partial F_T(L_{i-1})}{\partial\theta_u}\right)+q_{ij}\left(\frac{\partial^2 F_T(L_{i-1})}{\partial\theta_u\partial\theta_v}\right)\right]\left(\frac{\partial F_T(L_{i-1})}{\partial\theta_w}\right)\\
&+\frac{1}{\overline{F}_T(L_{i-1})}\left[\frac{\partial^3 \left[G(j,L_i)-G(j,L_{i-1})\right]}{\partial\theta_u\partial\theta_v\partial\theta_w}+\left(\frac{\partial^2 q_{ij}}{\partial\theta_v\partial\theta_w}\right)\left(\frac{\partial F_T(L_{i-1})}{\partial\theta_u}\right)+\left(\frac{\partial q_{ij}}{\partial\theta_v}\right)\left(\frac{\partial^2 F_T(L_{i-1})}{\partial\theta_u\partial_w}\right)\right.\\
&~~~~~~~~~~~~~~~~~~\left.+q_{ij}\left(\frac{\partial^3 F_T(L_{i-1})}{\partial\theta_u\partial\theta_v\partial\theta_w}\right)+\left(\frac{\partial q_{ij}}{\partial\theta_w}\right)\left(\frac{\partial^2 F_T(L_{i-1})}{\partial\theta_u\partial\theta_v}\right)\right].
    \end{align*}
\normalsize      Under the regularity conditions, it is assumed that the third-order derivatives of $G(j, L_i)$ and $F_T(L_i)$ are bounded. Therefore, we define
  \begin{align}\label{d3f}
        A_{uvw}=\max_{1\leq i\leq M}\sup_{\boldsymbol{\theta}\in\boldsymbol{\Theta}_0}\left|\frac{\partial^3F_T{(L_i)}}{\partial\theta_u\partial\theta_v\partial\theta_w}\right|
    \end{align}
    and
    \begin{align}\label{d3g}
    B_{uvw}=\max_{1\leq i\leq M} \max_{1\leq j\leq J}\sup_{\boldsymbol{\theta}\in\boldsymbol{\Theta}_0}\left|\frac{\partial^3 G{(j,L_i)}}{\partial\theta_u\partial\theta_v\partial\theta_w}\right|,
    \end{align}
        for $u=1,2,\ldots,s$, $v=1,2,\ldots,s$ and $w=1,2,\ldots,s$. Then, by using (\ref{F}-\ref{d3g}), we have
          \begin{align*}
        \left|\frac{\partial^3 q_{ij}}{\partial\theta_u\partial\theta_v\partial\theta_w}\right|\leq&2\kappa^3[2B_{u}+A_u]A_vA_w+\kappa^2[2B_{u}+C_{u}A_u]A_{vw}+\boldsymbol{\zeta}^2[2B_{uw}+A_{uw}+C_{w}A_u]A_v\\
&+\kappa^2[2B_{uv}+C_{v}A_u+A_{uv}]A_w+\kappa[3B_{uvw}+C_{vw}A_u+C_{v}A_{uw}+A_{uvw}+C_{w}A_{uv}]\\
&=C_{uvw}, ~~\text{say}
    \end{align*}
    This shows that the third-order derivative of $q_{ij}$ is bounded. We know that $q_{i}=\sum_{j=1}^Jq_{ij}$. Therefore, we get
    $ \left|\frac{\partial q_{i}}{\partial\theta_u}\right|\leq JC_{u}$,
     $\left|\frac{\partial^2 q_{i}}{\partial\theta_u\partial\theta_v}\right|\leq JC_{uv}$
and
     $\left|\frac{\partial^3 q_{i}}{\partial\theta_u\partial\theta_v\partial\theta_w}\right|\leq JC_{uvw}$.
 This shows that the first, second, and third-order derivatives of $q_{i}$ are also bounded. 
\end{proof}


\begin{lemma}\label{l4}
    Under the regularity conditions (I)-(VII), the ratio $\frac{N_i}{n}$ converges in probability to a finite number $b_i=\prod_{l=1}^{j-1}(1-p_l)(1-q_l)$, as $n\rightarrow \infty$, for $i=1,\ldots, M$.
\end{lemma}

A proof of the above lemma follows from \citep{sonal_17}.

\begin{lemma}\label{l_new}
    If the regularity conditions  (I)-(VII) hold, then
    \begin{enumerate}[(i)]
\item $I(\boldsymbol{\theta}\ | \ \boldsymbol{\boldsymbol{\zeta}})$ is finite for all $\boldsymbol{\theta}\in \boldsymbol{\Theta}_0$.
\item $I(\boldsymbol{\theta}\ | \ \boldsymbol{\boldsymbol{\zeta}})$ is positive definite for all $\boldsymbol{\theta}\in \boldsymbol{\Theta}_0$.
    \end{enumerate}
\end{lemma}
\begin{proof}
\textbf{Proof of part (i): } Using (\ref{Inform_uv}), we have 
    \begin{align*}
        I_{uv}(\boldsymbol{\theta}\ | \ \boldsymbol{\boldsymbol{\zeta}})=\sum_{i=1}^M\left[\sum_{j=1}^J\frac{E[N_{i}]} {q_{ij}}\frac{\partial q_{ij}}{\partial\theta_u}\frac{\partial q_{ij}}{\partial \theta_v}+\frac{E[N_{i}]}{(1-q_{i})}\frac{\partial q_i}{\partial\theta_u}\frac{\partial q_i}{\partial\theta_v}\right],
    \end{align*}
    for $u,v=1,2\cdots,s$. \\  
    \indent Note that for fixed n, $E[N_i]<n$, for $i=1,\ldots, M$. Since $0<q_{ij}<1$,  $0<\frac{1}{q_{ij}}<\infty$, $\forall \boldsymbol{\theta}\in \boldsymbol{\Theta}_0 $. Also, since $0<q_i<1$, $0<\frac{1}{q_i}< \infty$, $\forall \boldsymbol{\theta}\in \boldsymbol{\Theta}_0$. From Lemma \ref{l3}, we also have that the first-order derivatives of $q_{ij}$ and $q_i$ are bounded, for $i=1,2,\ldots,M$ and $j=1,2,\ldots,J$. Thus, $I_{uv}(\boldsymbol{\theta}\ | \ \boldsymbol{\boldsymbol{\zeta}})<\infty$. Hence, the proof. 

\textbf{Proof of part (ii): }  By definition of $I(\boldsymbol {\theta}|{\boldsymbol{\zeta}})$, it is a symmetric matrix. For any vector $a(\boldsymbol{\theta})\neq \boldsymbol{0}$, we get
\begin{align}\label{pa}
    a(\boldsymbol{\theta})^T I(\boldsymbol{\theta}\ | \ \boldsymbol{\boldsymbol{\zeta}})a(\boldsymbol{\theta})=\sum_{i=1}^M\left[\sum_{j=1}^J\frac{E[N_i]}{q_{ij}}\left\{a(\boldsymbol{\theta})^T\nabla_{\boldsymbol{\theta}}(q_{ij})\right\}^2+\frac{E[N_i]}{(1-q_i)}\left\{a(\boldsymbol{\theta})^T\nabla(q_i)\right\}^2\right].
\end{align}
Since $0<q_i<1$, $0<q_{ij}<1$ and $E[N_i]>0$, then from the equation (\ref{pa}), we can say that $I(\boldsymbol{\theta}\ | \ \boldsymbol{\boldsymbol{\zeta}})$ is non-negative definite matrix. We know prove that it is a positive definite matrix. We use the method of contradiction for this purpose. Let us consider that there exists a vector $a'(\boldsymbol{\theta})\neq 0$ such that 
\begin{align}\label{pa1}
    a'(\boldsymbol{\theta})^T I(\boldsymbol{\theta}\ | \ \boldsymbol{\boldsymbol{\zeta}})a'(\boldsymbol{\theta})=\sum_{i=1}^M\left[\sum_{j=1}^J\frac{E[N_i]}{q_{ij}}\left\{a'(\boldsymbol{\theta})^T\nabla_{\boldsymbol{\theta}}(q_{ij})\right\}^2+\frac{E[N_i]}{(1-q_i)}\left\{a'(\boldsymbol{\theta})^T\nabla(q_i)\right\}^2\right].
\end{align}
Now the equation \ref{pa1} holds if $a'(\boldsymbol{\theta})^T\nabla_{\boldsymbol{\theta}}(q_{ij})=0$ and $a'(\boldsymbol{\theta})^T\nabla(q_i)=0$, for all $i=1,\ldots,M$ and $j=1,\ldots,J$. However, this implies that $\nabla \boldsymbol{q}$ is a rank-deficient matrix, which contradicts the regularity conditions V(b). Hence $I(\boldsymbol{\theta}\ | \ \boldsymbol{\boldsymbol{\zeta}})$ is a positive definite matrix for all $\boldsymbol{\theta}\in\boldsymbol{\Theta}$.
\end{proof}
\begin{lemma}\label{l5}
   Suppose that the regularity conditions (I)-(VII) hold. Let $\boldsymbol{\Theta}_0$ be an open subset in the parameter space $\boldsymbol{\Theta}$ containing the true parameter $\boldsymbol{\theta}^0$. Then, all third derivatives $\frac{\partial^3 l(\boldsymbol{\theta})}{\partial\theta_u\partial\theta_v\partial\theta_w}$ exist. Furthermore, there exists a bound $K_{uvw}(\mathcal{D}_M)$ such that $\left|\frac{\partial^3 l(\boldsymbol{\theta})}{\partial\theta_u\partial\theta_v\partial\theta_w}\right|\leq K_{uvw}(\mathcal{D}_M)$, with $E[K_{uvw}]=nk_{uvw}$, where $k_{uvw}$ is finite for $u,v,w=1,2,\ldots,s$.
\end{lemma}
\begin{proof}
    Differentiating both sides of (\ref{lemma_2_dl}) with respect to $\theta_w$, we get
\small\allowdisplaybreaks    \begin{align}
          \frac{\partial^3 l(\boldsymbol{\theta})}{\partial\theta_u\partial\theta_v\partial\theta_w}&=\sum_{i=1}^M\sum_{k=1}^{n_i}\left[\sum_{j=1}^J\left(\frac{\delta_{ikj}} {q_{ij}}\right)\left(\frac{\partial^3 q_{ij}}{\partial\theta_u\partial\theta_v\partial\theta_w}\right)-\frac{(1-{\delta_{ik+}})}{(1-q_{i})}\left(\frac{\partial^3 q_{i}}{\partial\theta_u\partial\theta_v\partial\theta_w}\right)\right]\nonumber\\
          &-\sum_{i=1}^M\sum_{k=1}^{n_i}\sum_{j=1}^J\left(\frac{\delta_{ikj}} {q_{ij}^2}\right)\left[\left(\frac{\partial^2 q_{ij}}{\partial\theta_u\partial\theta_v}\right)\left(\frac{\partial q_{ij}}{\partial\theta_w}\right)+\left(\frac{\partial^2 q_{ij}}{\partial\theta_u\partial\theta_w}\right)\left(\frac{\partial q_{ij}}{\partial\theta_v}\right)+\left(\frac{\partial^2 q_{ij}}{\partial\partial\theta_v\theta_w}\right)\left(\frac{\partial q_{ij}}{\partial\theta_u}\right)\right]\nonumber\\
          &-\sum_{i=1}^M\sum_{k=1}^{n_i}\frac{(1-{\delta_{ik+}})}{(1-q_{i})^2}\left[\left(\frac{\partial^2 q_{i}}{\partial\theta_u\partial\theta_v}\right)\left(\frac{\partial q_{i}}{\partial\theta_w}\right)+\left(\frac{\partial^2 q_{i}}{\partial\theta_u\partial\theta_w}\right)\left(\frac{\partial q_{i}}{\partial\theta_v}\right)+\left(\frac{\partial^2 q_{i}}{\partial\theta_v\partial\theta_w}\right)\left(\frac{\partial q_{i}}{\partial\theta_u}\right)\right]\nonumber\\
     &+2\sum_{i=1}^M\sum_{k=1}^{n_i}\left[\sum_{j=1}^J\left(\frac{\delta_{ikj}} {q_{ij}^3}\right)\left(\frac{\partial q_{ij}}{\partial\theta_u}\right)\left(\frac{\partial q_{ij}}{\partial\theta_v}\right)\left(\frac{\partial q_{ij}}{\partial\theta_w}\right)-\frac{(n_i-{\delta_{ik+}})}{(1-q_{i})^3}\left(\frac{\partial q_{i}}{\partial\theta_u}\right)\left(\frac{\partial q_{i}}{\partial\theta_v}\right)\left(\frac{\partial q_{i}}{\partial\theta_w}\right)\right].\nonumber
    \end{align}
\normalsize We now define $$\gamma_1=\max_{1\leq i\leq M}\max_{1\leq j\leq J} \left(\frac{1}{q_{ij}}\right)$$ and $$\gamma_2=\max_{1\leq i\leq M}\left(\frac{1}{1-q_i}\right).$$
Thus, we have
\begin{align*}
    \left|\frac{\partial^3 l(\boldsymbol{\theta})}{\partial\theta_u\partial\theta_v\partial\theta_w}\right|\leq \sum_{i=1}^M n_i&\left[(\gamma_1 +\gamma_2)JC_{uvw}+(\gamma_1^2 +\gamma_2^2)J^2(C_{uv}C_w+C_{vw}C_u+C_{uw}C_v)\right.\\
    &\left.+2(\gamma_1^3J +\gamma_2^3J^3)C_{u}C_vC_w\right]=\sum_{i=1}^M n_ib_{uvw},\,\, \text{say.}
\end{align*}
 Let $\sum_{i=1}^M n_ib_{uvw}=K_{uvw}$. Now $E\left[K_{uvw}\right]=nk_{uvw}$, where $k_{uvw}=b_{uvw}\sum_{i=1}^M b_i$. It is easy to see that $k_{uvw}<\infty$. Hence, the proof.
\end{proof}\\
\textbf{Proof of Theorem 2: }
\begin{proof}
\indent 
\textbf The second part of the theorem can be proved following the steps as in \citep[][Theorem 5.1, page 463]{lehmann_casella_statistical_inference}. For the first part, it is enough to prove the following results: 
\begin{enumerate}[(a)]
    \item $\frac{1}{n}\frac{\partial l(\boldsymbol{\theta})}{\partial\theta_u}\xrightarrow{ P}$ 0, as $n\rightarrow\infty$.
    \item $\frac{1}{n}\frac{\partial^2 l(\boldsymbol{\theta})}{\partial\theta_u\partial\theta_v}\xrightarrow{P}-I^1_{uv}(\boldsymbol{\theta}\ | \ \boldsymbol{\boldsymbol{\zeta}})=-\sum\limits_{i=1}^M b_i\sigma_{i,uv}^2$, as $n\rightarrow\infty$,\\ where $\sigma_{i,uv}^2=\left[\frac{1}{q_{ij}}\left(\frac{\partial q_{ij}}{\partial\theta_u}\right)\left(\frac{\partial q_{ij}}{\partial\theta_v}\right)+\left(\frac{1}{1-q_i}\right)\left(\frac{\partial q_{i}}{\partial\theta_u}\right)\left(\frac{\partial q_{i}}{\partial\theta_v}\right)\right]$
\end{enumerate}

Now, we have
\begin{align*}
  \frac{1}{n}  \frac{\partial l (\boldsymbol{\theta})}{\partial\theta_u}=\sum_{i=1}^M \frac{n_i}{n}\left[\frac{1}{n_i}\sum_{k=1}^{n_i}U_{ik,u}\right],
\end{align*}
where $U_{ik,u}=\sum\limits_{j=1}^J\left(\frac{\delta_{ikj}} {q_{ij}}\right)\left(\frac{\partial q_{ij}}{\partial\theta_u}\right)-\frac{(1-{\delta_{ik+}})}{(1-q_{i})}\left(\frac{\partial q_{i}}{\partial\theta_u}\right)$. For given $n_i$, $U_{ik,u}$s are i.i.d. and $E[U_{ik,u}\ | \ n_i]=0$, for $k=1,\ldots,n_i$. Now, by the Weak Law of Large Numbers (WLLN), for given $n_i$, $\sum\limits_{k=1}^{n_i}U_{ik,u}/n_i\xrightarrow{P}0$, as $n\rightarrow\infty $. Now, by Lemma \ref{l3}, it follows that 
$\frac{1}{n} \frac{\partial l (\boldsymbol{\theta})}{\partial\theta_u}\xrightarrow{P}0$ as $n\rightarrow\infty$.\\
\indent Similarly,  we have
\begin{align*}
  \frac{1}{n}  \frac{\partial l^2 (\boldsymbol{\theta})}{\partial\theta_u\partial\theta_v}=\sum_{i=1}^M \frac{n_i}{n}\left[\frac{1}{n_i}\sum_{k=1}^{n_i}V_{ik,uv}\right],
\end{align*}
where $V_{ik,uv}=\left[\sum\limits_{j=1}^J\left(\frac{\delta_{ikj}} {q_{ij}}\right)\left(\frac{\partial^2 q_{ij}}{\partial\theta_u\partial\theta_v}\right)-\frac{(1-{\delta_{ik+}})}{(1-q_{i})}\left(\frac{\partial^2 q_{i}}{\partial\theta_u\partial\theta_v}\right)\right]-\left[\sum\limits_{j=1}^J\left(\frac{\delta_{ikj}} {q_{ij}^2}\right)\left(\frac{\partial q_{ij}}{\partial\theta_u}\right)\left(\frac{\partial q_{ij}}{\partial\theta_v}\right)+\frac{(1-{\delta_{ik+}})}{(1-q_{i})^2}\right.$\\ 
$\left.\left(\frac{\partial q_{i}}{\partial\theta_u}\right)\left(\frac{\partial q_{i}}{\partial\theta_v}\right)\right].$ For given $n_i$, $V_{ik,uv}$s are i.i.d. and $E[V_{ik,uv}\ | \ n_i]=-\left[\frac{1}{q_{ij}}\left(\frac{\partial q_{ij}}{\partial\theta_u}\right)\left(\frac{\partial q_{ij}}{\partial\theta_v}\right)+\left(\frac{1}{1-q_i}\right)\right.$\\$\left.\left(\frac{\partial q_{i}}{\partial\theta_u}\right)\left(\frac{\partial q_{i}}{\partial\theta_v}\right)\right]$,
for $k=1,\ldots,n_i$. Now, by the WLLN, for a given $n_i$, $$\sum\limits_{k=1}^{n_i}\frac{V_{ik,uv}}{n_i}\xrightarrow{P}-\left[\frac{1}{q_{ij}}\left(\frac{\partial q_{ij}}{\partial\theta_u}\right)\left(\frac{\partial q_{ij}}{\partial\theta_v}\right)+\left(\frac{1}{1-q_i}\right)\left(\frac{\partial q_{i}}{\partial\theta_u}\right)\left(\frac{\partial q_{i}}{\partial\theta_v}\right)\right]=-\sigma^2_{i,uv},$$ as $n\rightarrow\infty $. Now, by Lemma \ref{l4}, we get the desired result.
\end{proof}\\
\section{Computation of $\nabla_{\boldsymbol{\theta}}(q_{ij})$ and $\nabla_{\boldsymbol{\theta}}(q_{i})$}\label{computation}
From (\ref{qij}) and (\ref{qi}), we have
\begin{align*}
  \log(q_{ij})&=\log\left(\overline{G}\left(j,L_{i-1}|\boldsymbol{\theta}\right)-\overline{G}\left(j,L_{i}|\boldsymbol{\theta}\right)\right)-\log\left(\overline{F}_T\left(L_{i-1}|\boldsymbol{\theta}\right)\right)
\end{align*} 
and
\begin{align*}
\log(1-q_{i})=\log\left(\overline{F}_T\left(L_{i}|\boldsymbol{\theta}\right)\right)-\log\left(\overline{F}_T\left(L_{i-1}|\boldsymbol{\theta}\right)\right),
\end{align*}
for $j=1,\ldots,J$ and $i=1,\ldots,M$. Now, using equations \ref{sub} and \ref{sub1}, the term $log(q_{ij})$ can be written as
   \begin{align*}
\log(q_{ij})=&\log\left(\overline{F}_T\left(L_{i-1}|\boldsymbol{\theta}\right)-\overline{F}_T\left(L_i|\boldsymbol{\theta}\right)\right)-\psi_j\left(\boldsymbol{\theta}\right)-\log\left(\overline{F}_T\left(L_{i-1}|\boldsymbol{\theta}\right)\right),
\end{align*}
where $\psi_j({\boldsymbol{\theta}})=\log\left(\sum_{j'=1}^J\left(\frac{\eta_j}{\eta_{j'}}\right)^\gamma\right)$. Thus, we get
\allowdisplaybreaks\bea
        \frac{1}{q_{ij}}\left(\frac{\partial q_{ij}}{\partial \theta_u}\right)&=&\frac{1}{\overline{F}_T\left(L_{i-1}|\boldsymbol{\theta}\right)-\overline{F}_T\left(L_i|\boldsymbol{\theta}\right)}\left[\left(\frac{\partial \overline{F}_T\left(L_{i-1}|\boldsymbol{\theta}\right)}{\partial\theta_u}\right)-\left(\frac{\partial \overline{F}_T\left(L_{i}|\boldsymbol{\theta}\right)}{\partial\theta_u}\right)\right]-\left(\frac{\partial\psi_j\left(\boldsymbol{\theta}\right)}{\partial \theta_u}\right)\nonumber\\
        &-&\frac{1}{ \overline{F}_T\left(L_{i-1}|\boldsymbol{\theta}\right)}\left(\frac{\partial \overline{F}_T\left(L_{i-1}|\boldsymbol{\theta}\right)}{\partial\theta_u}\right),\nonumber
    \eea
and 
\begin{align*}
        \left(\frac{1}{1-q_{i}}\right)\left(\frac{\partial q_{i}}{\partial \theta_u}\right)=\frac{1}{\overline{F}_T\left(L_{i-1}|\boldsymbol{\theta}\right)}\left(\frac{\partial \overline{F}_T\left(L_{i-1}|\boldsymbol{\theta}\right)}{\partial\theta_u}\right)-\frac{1}{\overline{F}_T(L_{i}|\boldsymbol{\theta})}\left(\frac{\partial \overline{F}_T(L_{i}|\boldsymbol{\theta})}{\partial\theta_u}\right),
\end{align*}
for $u=1,\ldots,s$.\\
\indent Thus, it is straightforward to obtain $S^2$, since we readily have
\begin{align*}
\frac{\partial{\overline{F}_T(t|\boldsymbol{\theta})}}{\partial \eta_j}=\frac{\gamma}{\eta_j}\left[1+\nu \sum_{j=1}^J\left(\frac{t}{\eta_j}\right)^\gamma\right]^{-\frac{1}{\nu}-1}\left(\frac{t}{\eta_j}\right)^{\gamma},   
\end{align*}
\begin{align*}
\frac{\partial{\overline{F}_T(t|\boldsymbol{\theta})}}{\partial \gamma}=-\left[1+\nu \sum_{j=1}^J\left(\frac{t}{\eta_j}\right)^\gamma\right]^{-\frac{1}{\nu}-1}\sum_{j=1}^J\left(\frac{t}{\eta_j}\right)^\gamma \log\left(\frac{t}{\eta_j}\right),    
\end{align*}
and 
\begin{align*}
\frac{\partial{\overline{F}_T(t|\boldsymbol{\theta})}}{\partial \nu}=\left[1+\nu\sum_{j=1}^J\left(\frac{t}{\eta_j}\right)^{\gamma}\right]^{-\frac{1}{\nu}}\left[\frac{1}{\nu^2}\log\left(1+\nu \sum_{j=1}^J\left(\frac{t}{\eta_j}\right)^\gamma\right)-\frac{1}{\nu}\left(\frac{\sum_{j=1}^J\left(\frac{t}{\eta_j}\right)^\gamma}{1+\nu \sum_{j=1}^J\left(\frac{t}{\eta_j}\right)^\gamma}\right)\right]
\end{align*}
for the dependent competing risks model in (\ref{sub}). Similarly, for the independent competing risks model in (\ref{sub1}), we have 
\begin{align*}
\frac{\partial{\overline{F}_T(t|\boldsymbol{\theta})}}{\partial \eta_j}=\frac{\gamma}{\eta_j}\left(\frac{t}{\eta_j}\right)^{\gamma}\exp\left(- \sum_{j=1}^J\left(\frac{t}{\eta_j}\right)^\gamma\right),
\end{align*}
and
\begin{align*}
\frac{\partial{\overline{F}_T(t|\boldsymbol{\theta})}}{\partial \gamma}=-\exp\left(-\sum_{j=1}^J\left(\frac{t}{\eta_j}\right)^\gamma\right)\sum_{j=1}^J\left(\frac{t}{\eta_j}\right)^\gamma \log\left(\frac{t}{\eta_j}\right).
\end{align*}
Furthermore, we also note that
\begin{align*}
\frac{\partial{\psi_j(\boldsymbol{\theta})}}{\partial \eta_j}=\frac{\left(\frac{\gamma}{\eta_j}\right)\sum_{j'\neq j}\left(\frac{\eta_j}{\eta_{j'}}\right)^\gamma}{\sum_{j'=1}^J\left(\frac{\eta_j}{\eta_{j'}}\right)^\gamma},
\end{align*}
\begin{align*}
\frac{\partial{\psi_j(\boldsymbol{\theta})}}{\partial \eta_{j'}}=-\frac{\left(\frac{\gamma}{\eta_{j'}}\right)\sum_{j'\neq j}\left(\frac{\eta_j}{\eta_{j'}}\right)^\gamma}{\sum_{j'=1}^J\left(\frac{\eta_j}{\eta_{j'}}\right)^\gamma},\,\, \text{for}\,\ j' \neq j, 
\end{align*}
and
\begin{align*}
\frac{\partial{\psi_j(\boldsymbol{\theta})}}{\partial \gamma}=\frac{\sum_{j=1}^j\left(\frac{\eta_j}{\eta_{j'}}\right)^\gamma \log\left(\frac{\eta_j}{\eta_{j'}}\right)}{\sum_{j'=1}^J\left(\frac{\eta_j}{\eta_{j'}}\right)^\gamma}.
\end{align*}

\section{Proof of Theorem \ref{theorem5}} \label{theoremapp}

\begin{proof}
\textbf{Proof of part (i): }  Let $\boldsymbol{\boldsymbol{\zeta}}_1=(M, h,p)$ and $\boldsymbol{\boldsymbol{\zeta}}_2=(M+1,h,p)$.  Then, it is enough to show that $\phi(\boldsymbol{\boldsymbol{\zeta}}_2)\leq \phi(\boldsymbol{\boldsymbol{\zeta}}_1)$.
    \begin{align*}
    &I(\boldsymbol{\theta}\ | \ \boldsymbol{\boldsymbol{\zeta}}_2)-I(\boldsymbol{\theta}\ | \ \boldsymbol{\boldsymbol{\zeta}}_1)\\
  =&\sum_{j=1}^J\frac{E[N_{M+1}]} {q_{(M+1)j}}\nabla_{\boldsymbol{\theta}}(q_{(M+1)j})\nabla_{\boldsymbol{\theta}}(q_{(M+1)j})^T+\frac{E[N_{M+1}]}{(1-q_{M+1})}\nabla_{\boldsymbol{\theta}}(q_{(M+1)})\nabla_{\boldsymbol{\theta}}(q_{(M+1)})^T.
\end{align*}
Now, for any  $\boldsymbol{a}(\boldsymbol{\theta})\neq 0$, we have
\begin{align*}
   & \boldsymbol{a}(\boldsymbol{\theta})^T[I(\boldsymbol{\theta}\ | \ \boldsymbol{\boldsymbol{\zeta}}_2)-I(\boldsymbol{\theta}\ | \ \boldsymbol{\boldsymbol{\zeta}}_1)]\boldsymbol{a}(\boldsymbol{\theta})\\
    =&\sum_{j=1}^J\frac{E[N_{M+1}]} {q_{(M+1)j}}\left[\boldsymbol{a}(\boldsymbol{\theta})^T\nabla_{\boldsymbol{\theta}}q_{(M+1)j}\right]^2+\frac{E[N_{M+1}]}{(1-q_{M+1})}\left[\boldsymbol{a}(\boldsymbol{\theta})^T\nabla_{\boldsymbol{\theta}}q_{(M+1)}\right]^2
\end{align*}
This shows that $I(\boldsymbol{\theta}\ | \ \boldsymbol{\boldsymbol{\zeta}}_2)-I(\boldsymbol{\theta}\ | \ \boldsymbol{\boldsymbol{\zeta}}_1)$ is a non-negative definite matrix. This implies that $I^{-1}(\boldsymbol{\theta}\ | \ \boldsymbol{\boldsymbol{\zeta}}_1)-I^{-1}(\boldsymbol{\theta}\ | \ \boldsymbol{\boldsymbol{\zeta}}_2)$ is also non-negative definite matrix. Therefore, we have
\begin{align*}
    &\boldsymbol{a}(\boldsymbol{\theta})^T\left[I^{-1}(\boldsymbol{\theta}\ | \ \boldsymbol{\boldsymbol{\zeta}}_1)-I^{-1}(\boldsymbol{\theta}\ | \ \boldsymbol{\boldsymbol{\zeta}}_2)\right]\boldsymbol{a}(\boldsymbol{\theta})\geq 0\\
    \implies& \boldsymbol{a}(\boldsymbol{\theta})^TI^{-1}(\boldsymbol{\theta}\ | \ \boldsymbol{\boldsymbol{\zeta}}_1)\boldsymbol{a}(\boldsymbol{\theta})-\boldsymbol{a}(\boldsymbol{\theta})^TI^{-1}(\boldsymbol{\theta}\ | \ \boldsymbol{\boldsymbol{\zeta}}_2)\boldsymbol{a}(\boldsymbol{\theta})\geq 0\\
    \implies&\phi( \boldsymbol{\boldsymbol{\zeta}}_2)\leq \phi(\boldsymbol{\boldsymbol{\zeta}}_1).
\end{align*}
\textbf{Proof of part (ii):} Let  $\boldsymbol{\boldsymbol{\zeta}}_1=(M, h,p_1)$ and $\boldsymbol{\boldsymbol{\zeta}}_2=(M,h,p_2)$, where $p_1<p_2$.  Then it is enough to prove that $\phi(\boldsymbol{\boldsymbol{\zeta}}_1)\leq \phi(\boldsymbol{\boldsymbol{\zeta}}_2)$.
\begin{align*}
   & I(\boldsymbol{\theta}\ | \ \boldsymbol{\boldsymbol{\zeta}}_1)-I(\boldsymbol{\theta}\ |\ \boldsymbol{\boldsymbol{\zeta}_2})\\
    =&\sum_{i=1}^M\left[\sum_{j=1}^J\frac{\left(E[N_i^1]-E[N_{i}^2]\right)} {q_{ij}}\nabla_{\boldsymbol{\theta}}(q_{ij})\nabla_{\boldsymbol{\theta}}(q_{ij})^T+\frac{\left(E[N_i^1]-E[N_{i}^2]\right)}{(1-q_{i})}\nabla_{\boldsymbol{\theta}}(q_{i})\nabla_{\boldsymbol{\theta}}(q_{i})^T\right],
\end{align*}
where $E[N_i^k]=n (1-p_k)^{i-1}\overline{F}_T(L_{i-1}|\boldsymbol{\theta})$, for $i=1,\ldots,M$ and $k=1,2$.
Now, for any $\boldsymbol{a}(\boldsymbol{\theta})\neq 0$, we have
\begin{align*}
   & \boldsymbol{a}(\boldsymbol{\theta})^T[I(\boldsymbol{\theta}\ | \ \boldsymbol{\boldsymbol{\zeta}}_1)-I(\boldsymbol{\theta}\ | \ \boldsymbol{\boldsymbol{\zeta}_2})]\boldsymbol{a}(\boldsymbol{\theta})\\
    =&\sum_{i=1}^M\left[\sum_{j=1}^J\frac{\left(E[N_i^1]-E[N_{i}^2]\right)} {q_{ij}}\left[\boldsymbol{a}(\boldsymbol{\theta})^T\nabla_{\boldsymbol{\theta}}(q_{ij})\right]^2+\frac{\left(E[N_i^1]-E[N_{i}^2]\right)}{(1-q_{i})}\left[\boldsymbol{a}(\boldsymbol{\theta})^T\nabla_{\boldsymbol{\theta}}(q_{i})\right]^2\right]
\end{align*}
Since 
\begin{align*}
    E[N_i^1]-E[N_i^2]&=n (1-p_1)^{i-1} \overline{F}_T(L_{i-1}|\boldsymbol{\theta})-n(1-p_2)^{i-1} \overline{F}_T(L_{i-1}|\boldsymbol{\theta})\\
    &=n \overline{F}_T(L_{i-1}|\boldsymbol{\theta})\left[(1-p_1)^{i-1}-(1-p_2)^{i-1}\right]\geq0
\end{align*}
This shows that $I(\boldsymbol{\theta}\ | \ \boldsymbol{\boldsymbol{\zeta}}_1)-I(\boldsymbol{\theta}\ | \ \boldsymbol{\boldsymbol{\zeta}}_2)$ is a non-negative definite matrix. This implies that $I^{-1}(\boldsymbol{\theta}\ | \ \boldsymbol{\boldsymbol{\zeta}}_2)-I^{-1}(\boldsymbol{\theta}\ | \ \boldsymbol{\boldsymbol{\zeta}}_1)$ is also non-negative definite matrix. Therefore, we have
\begin{align*}
    &\boldsymbol{a}(\boldsymbol{\theta})^T\left[I^{-1}(\boldsymbol{\theta}\ | \ \boldsymbol{\boldsymbol{\zeta}}_2)-I^{-1}(\boldsymbol{\theta}\ | \ \boldsymbol{\boldsymbol{\zeta}}_1)\right]\boldsymbol{a}(\boldsymbol{\theta})\geq 0\\
    \implies& \boldsymbol{a}(\boldsymbol{\theta})^TI^{-1}(\boldsymbol{\theta}\ | \ \boldsymbol{\boldsymbol{\zeta}}_2)\boldsymbol{a}(\boldsymbol{\theta})-\boldsymbol{a}(\boldsymbol{\theta})^TI^{-1}(\boldsymbol{\theta}\ | \ \boldsymbol{\boldsymbol{\zeta}}_1)\boldsymbol{a}(\boldsymbol{\theta})\geq 0\\
    \implies&\phi( \boldsymbol{\boldsymbol{\zeta}}_1)\leq \phi(\boldsymbol{\boldsymbol{\zeta}}_2).
\end{align*}
\end{proof}

\end{document}